\documentclass[letterpaper,11pt]{article}

\usepackage{amsthm,amsmath,amssymb,amsfonts}
\usepackage[american]{babel}
\usepackage{graphicx}
\usepackage[url=false,giveninits=true,maxbibnames=99,style=alphabetic,maxalphanames=5]{biblatex}
\usepackage{hyperref}
\usepackage{xurl}
\hypersetup{breaklinks=true}
\usepackage{mathtools}
\usepackage{enumitem}
\usepackage{csquotes}
\usepackage[noend]{algpseudocode}
\usepackage{algorithm}
\usepackage{xparse}
\usepackage{xspace}
\usepackage{color}
\usepackage{caption}
\usepackage{subcaption}
\usepackage{fullpage}
\usepackage{placeins}
\usepackage{xcolor}
\usepackage{makecell}
\usepackage{thmtools}
\usepackage{thm-restate}
\usepackage{cleveref}
\usepackage{tikz}
\usepackage{tabularx}
\usepackage{diagbox}

\newcommand{\gray}[1]{\textcolor{gray}{{#1}}}

\mathchardef\mhyphen="2D 

\newcommand\newmathabbrev[2]{\newcommand{#1}{\ensuremath{#2}\xspace}}

\newcommand\cfont\mathrm
\newmathabbrev\p{\cfont{P}}
\newmathabbrev{\N}{\mathbb N}
\newmathabbrev\NP{\cfont{NP}}
\newmathabbrev\DTIME{\cfont{DTIME}}
\newmathabbrev\tSAT{3\cfont{\mhyphen{}SAT}}
\newmathabbrev\MA{\cfont{MA}}
\newmathabbrev\MAo{\cfont{MA}_1}
\newmathabbrev\AM{\cfont{AM}}
\newmathabbrev\NPDAG{\cfont{NP\mhyphen{}DAG}}
\newmathabbrev\QMADAG{\cfont{QMA\mhyphen{}DAG}}
\newmathabbrev\yes{\mathrm{yes}}
\newmathabbrev\no{\mathrm{no}}
\newmathabbrev\US{\cfont{US}}
\newmathabbrev\FP{\cfont{FP}}
\newmathabbrev\PP{\cfont{PP}}
\newmathabbrev\CeP{\cfont{C_=P}}
\newmathabbrev\coCeP{\cfont{coC_=P}}
\newmathabbrev\PH{\cfont{PH}}
\newmathabbrev\SAT{\cfont{SAT}}
\newmathabbrev\QSAT{\cfont{QSAT}}
\newmathabbrev\hQSAT{\mhyphen\QSAT}
\newmathabbrev\SPP{\cfont{SPP}}
\newmathabbrev\GapP{\cfont{GapP}}
\newmathabbrev\BQP{\cfont{BQP}}
\newmathabbrev\QP{\cfont{QP}}
\newmathabbrev\StoqMA{\cfont{StoqMA}}
\newmathabbrev\coNP{\cfont{coNP}}
\newmathabbrev\AzPP{\cfont{A_0PP}}
\newmathabbrev\QMA{\cfont{QMA}}
\newmathabbrev\QMAo{\cfont{QMA}_1}
\newmathabbrev\coQMA{\cfont{coQMA}}
\newmathabbrev\BPP{\cfont{BPP}}
\newmathabbrev\QCMA{\cfont{QCMA}}
\newmathabbrev\pNPlog{\p^{\NP[\log]}}
\newmathabbrev\pNP{\p^{\NP}}
\newmathabbrev\pNPtwo{\p^{\NP[2]}}
\newmathabbrev\pNPone{\p^{\NP[1]}}
\newmathabbrev\pParSAT{\p^{||\SAT}}
\newmathabbrev\pQMApar{\p^{||\QMA}}
\newmathabbrev\pCpar{\p^{||\C}}
\newmathabbrev\pStoqMApar{\p^{||\StoqMA}}
\newmathabbrev\pQMAlog{\p^{\QMA[\log]}}
\newmathabbrev\pClog{\p^{\textup{C}[\log]}}
\newmathabbrev\pC{\p^{\textup{C}}}
\newmathabbrev\QMASPACE{\cfont{QMASPACE}}
\newmathabbrev\pQMAtlog{\p^{\QMA(2)[\log]}}
\newmathabbrev\pStoqMAlog{\p^{\StoqMA[\log]}}
\newmathabbrev\pQMApt{\p^{\Vert\QMA(2)}}
\newmathabbrev\pQMA{\p^{\QMA}}
\newmathabbrev\SharpP{\cfont{\#P}}
\newmathabbrev\pSharP{\p^{\SharpP[1]}}
\newmathabbrev\PromisePP{\cfont{PromisePP}}
\newmathabbrev\lett{\le_\mathrm{tt}}
\newmathabbrev\YES{\mathsf{YES}}
\newmathabbrev\NO{\mathsf{NO}}
\newmathabbrev\PSPACE{\cfont{PSPACE}}
\newmathabbrev\IP{\cfont{IP}}
\newmathabbrev\POLY{\cfont{POLY}}
\newmathabbrev\DAG{\cfont{DAG}}
\newmathabbrev\StoqMADAG{\StoqMA\mhyphen\cfont{DAG}}
\newmathabbrev\CDAG{C\mhyphen\cfont{DAG}}
\newmathabbrev\CDAGf{C\mhyphen\cfont{DAG}_f}
\newmathabbrev\CDAGs{C\mhyphen\cfont{DAG}_s}
\newmathabbrev\CDAGd{C\mhyphen\cfont{DAG}_{d}}
\newmathabbrev\CDAGo{C\mhyphen\cfont{DAG}_1}
\newmathabbrev\LOGS{\cfont{LOGS}}
\newmathabbrev\TAUT{\cfont{TAUTOLOGY}}
\newmathabbrev\SBQP{\cfont{SBQP}}
\newmathabbrev\Fc{F_\coNP}
\newmathabbrev\Fa{F_\AzPP}
\newmathabbrev\GSCON{\cfont{GSCON}}
\newmathabbrev\GSCONexp{\GSCON_\cfont{exp}}
\newmathabbrev\QMAexp{\QMA_\cfont{exp}}
\newmathabbrev\UQMA{\cfont{UQMA}}
\newmathabbrev\R{\mathbb R}
\newmathabbrev\Trees{\cfont{TREES}}
\newmathabbrev\apxsim{\cfont{APX\mhyphen{}SIM}}
\newmathabbrev\AWPP{\cfont{AWPP}}
\newmathabbrev\X{\mathcal{X}}
\newmathabbrev\Y{\mathcal{Y}}
\newmathabbrev\C{\mathcal{C}}

\newmathabbrev\Z{\mathcal{Z}}
\newmathabbrev\ZZ{\mathbb{Z}}
\newmathabbrev\Hprop{H_\mathrm{prop}}
\newmathabbrev\Hin{H_\mathrm{in}}
\newmathabbrev\Piin{\Pi_\mathrm{in}}
\newmathabbrev\Hout{H_\mathrm{out}}
\newmathabbrev\Piout{\Pi_\mathrm{out}}
\newmathabbrev\Hstab{H_\mathrm{stab}}
\newmathabbrev\Lext{\L_\mathrm{ext}}
\newmathabbrev\BTWNP{\cfont{BTW}(\NP)}
\newmathabbrev\BSN{\cfont{BSN}}
\newmathabbrev\SN{\cfont{SN}}
\newmathabbrev\BD{\cfont{BD}}
\newmathabbrev\HYPERTREE{\cfont{NP\mhyphen{}HYPERTREE}}
\newmathabbrev\Hext{H_\mathrm{ext}}
\newmathabbrev\Hpropt{\tilde{H}_\mathrm{prop}}
\newmathabbrev\Hint{\tilde{H}_\mathrm{in}}
\newmathabbrev\Houtt{\tilde H_\mathrm{out}}
\newmathabbrev\EXP{\cfont{EXP}}
\newmathabbrev\A{\mathcal{A}}
\newmathabbrev\U{\mathcal{U}}

\renewcommand\L{\mathcal{L}}

\newcommand\K{\mathcal K}
\newmathabbrev\DAGSSAT{\DAGS(\SAT)}
\newmathabbrev\DAGS{\mathrm{DAGS}}
\newmathabbrev\DAGSNP{\DAGS(\NP)}

\newcommand\psiin{\psi_\mathrm{in}}
\newmathabbrev\AND{\cfont{AND}}

\newmathabbrev\STCONN{{S,T}\cfont{\mhyphen{}CONN}}
\newmathabbrev\CNF{\cfont{CNF}}
\newmathabbrev\NEXP{\cfont{NEXP}}
\newmathabbrev\NPSPACE{\cfont{NPSPACE}}
\newmathabbrev\QCMASPACE{\cfont{QCMASPACE}}
\newmathabbrev\BQPSPACE{\cfont{BQPSPACE}}
\newmathabbrev{\PCP}{\cfont{PCP}}
\newmathabbrev\BQUPSPACE{\cfont{BQ_UPSPACE}}
\newmathabbrev\QMAt{\QMA(2)}
\newmathabbrev\QMAtexp{\QMAt_{\exp}}
\newmathabbrev\MIP{\cfont{MIP}}
\newmathabbrev\MIPt{\MIP(2)}
\newmathabbrev\BellQMA{\cfont{BellQMA}}
\newmathabbrev\BellQMAt{\BellQMA(2)}
\newmathabbrev\BellQMAtexp{\BellQMAt_{\exp}}

\protected\def\verythinspace{%
  \ifmmode
    \mskip0.5\thinmuskip
  \else
    \ifhmode
      \kern0.08334em
    \fi
  \fi
}

\newcommand{\bin}{\{0,1\}}

\newcommand{\CC}{\mathbb C}

\newcommand{\be}{\begin{equation}}
\newcommand{\ee}{\end{equation}}

\newcommand{\CNOT}{\mathrm{CNOT}}

\renewcommand{\epsilon}{\varepsilon}

\newcommand{\Null}{\mathcal{N}}
\DeclareMathOperator{\Image}{Im}

\newcommand\lmin{\lambda_{\mathrm{min}}}

\newcommand{\set}[1]{{\left\{#1\right\}}}    

\DeclareMathOperator{\poly}{poly}

\DeclareMathOperator{\Span}{Span}

\DeclarePairedDelimiter\bra{\langle}{\rvert}
\DeclarePairedDelimiter\ket{\lvert}{\rangle}

\DeclarePairedDelimiter\abs{\lvert}{\rvert}
\DeclarePairedDelimiter\norm{\lVert}{\rVert}

\DeclarePairedDelimiterX\braket[2]{\langle}{\rangle}{#1 \delimsize\vert #2}
\DeclarePairedDelimiterX\ketbra[2]{\lvert}{\rvert}{#1 \delimsize\rangle\delimsize\langle #2}

\newcommand{\braketb}[2]{\bra{#1}#2\ket{#1}}
\newcommand{\braketc}[1]{\braket{#1}{#1}}
\newcommand{\ketbraa}[1]{{#1 \renewcommand\ket\bra #1}}
\newcommand{\ketbrab}[1]{\ketbra{#1}{#1}}

\setlist[itemize]{noitemsep, topsep=0pt}
\setlist[enumerate]{noitemsep, topsep=0pt}

\declaretheorem[numberwithin=section]{theorem}
\declaretheorem[sibling=theorem]{corollary}
\declaretheorem[sibling=theorem]{lemma}

\declaretheorem[sibling=theorem,style=definition]{definition}
\declaretheorem[sibling=theorem,style=definition]{remark}

\newcommand{\hin}{H_{\textup{in}}}
\newcommand{\hprop}{H_{\textup{prop}}}

\newcommand{\spa}[1]{\mathcal{#1}}

\newcommand{\ayes}{A_{\textup{yes}}} 
\newcommand{\ano}{A_{\textup{no}}} 
\newcommand{\psihist}{\psi_{\textup{hist}}}

\newcommand{\HHcomp}{\mathcal{H}_{\mathrm{comp}}}
\newcommand{\clock}{\mathrm{clock}}
\newcommand{\HHclock}{\mathcal{H}_{\clock}}
\newcommand{\HHclocki}{\mathcal{H}_{\clock,1}}
\newcommand{\HHclockii}{\mathcal{H}_{\clock,2}}
\newcommand{\Piclock}{\Pi_\clock}
\newcommand{\PiclockN}{\Piclock^{(N)}}
\newcommand{\HclockN}{H_{\mathrm{clock}}^{(N)}}
\newcommand{\I}{\mathbb{I}}
\newcommand{\Hclockp}[1]{H_{\mathrm{clock},#1}}
\newcommand{\Hclock}{H_{\mathrm{clock}}}
\newcommand{\Hinit}{H_{\mathrm{init}}}
\newcommand{\Hdiag}{H_{\mathrm{diag}}}
\newcommand{\Pidiag}{\Pi_{\mathrm{diag}}}
\newcommand{\Htrans}{H_{\mathrm{trans}}}
\newcommand{\Hgate}{H_{\mathrm{gate}}}
\newcommand{\Hlink}{H_{\mathrm{link}}}
\newcommand{\Hend}{H_{\mathrm{end}}}

\newcommand{\LS}{\mathcal{L}}
\newcommand{\Ls}{\mathbf{L}}
\newcommand{\KK}{\mathcal{K}}

\makeatletter
\newcommand{\subalign}[1]{%
  \vcenter{%
    \Let@ \restore@math@cr \default@tag
    \baselineskip\fontdimen10 \scriptfont\tw@
    \advance\baselineskip\fontdimen12 \scriptfont\tw@
    \lineskip\thr@@\fontdimen8 \scriptfont\thr@@
    \lineskiplimit\lineskip
    \ialign{\hfil$\m@th\scriptstyle##$&$\m@th\scriptstyle{}##$\hfil\crcr
      #1\crcr
    }%
  }%
}
\NewDocumentCommand{\LeftComment}{s m}{%
  \Statex \IfBooleanF{#1}{\hspace*{\ALG@thistlm}}\(\triangleright\) #2}

\def\moverlay{\mathpalette\mov@rlay}
\def\mov@rlay#1#2{\leavevmode\vtop{%
   \baselineskip\z@skip \lineskiplimit-\maxdimen
   \ialign{\hfil$\m@th#1##$\hfil\cr#2\crcr}}}
\newcommand{\charfusion}[3][\mathord]{
    #1{\ifx#1\mathop\vphantom{#2}\fi
        \mathpalette\mov@rlay{#2\cr#3}
      }
    \ifx#1\mathop\expandafter\displaylimits\fi}

\makeatother

\algnewcommand{\LineComment}[1]{\State \(\triangleright\) #1}

\algblockdefx[ON]{Blk}{EndBlk}[1]
  {#1}
  {}

\makeatletter
\ifthenelse{\equal{\ALG@noend}{t}}%
  {\algtext*{EndBlk}}
  {}%
\makeatother

\newcolumntype{M}[1]{>{\centering\arraybackslash$}m{#1}<{$}}

\makeatletter
\newcommand\@rcolwidth{0.67em}
\newenvironment{rmatrix}{%
    \new@ifnextchar[\@rarray{\@rarray[\@rcolwidth]}%
}{%
    \endarray
}
\def\@rarray[#1]{\arraycolsep=0pt\array{*\c@MaxMatrixCols {M{#1}}}}
\makeatother

\newcommand{\splitatcommas}[1]{
  \begingroup
  \begingroup\lccode`~=`, \lowercase{\endgroup
    \edef~{\mathchar\the\mathcode`, \penalty0 \noexpand\hspace{0pt plus 1em}}%
  }\mathcode`,="8000 #1%
  \endgroup
}

\AtEveryBibitem{%
  \clearlist{language}%
}

\DeclareFontFamily{U}{matha}{\hyphenchar\font45}
\DeclareFontShape{U}{matha}{m}{n}{
      <5> <6> <7> <8> <9> <10> gen * matha
      <10.95> matha10 <12> <14.4> <17.28> <20.74> <24.88> matha12
      }{}
\DeclareSymbolFont{matha}{U}{matha}{m}{n}
\DeclareFontFamily{U}{mathx}{\hyphenchar\font45}
\DeclareFontShape{U}{mathx}{m}{n}{
      <5> <6> <7> <8> <9> <10>
      <10.95> <12> <14.4> <17.28> <20.74> <24.88>
      mathx10
      }{}
\DeclareSymbolFont{mathx}{U}{mathx}{m}{n}

\DeclareMathSymbol{\obot}         {2}{matha}{"6B}
\DeclareMathSymbol{\bigobot}       {1}{mathx}{"CB}

\newcommand\scalemath[2]{\scalebox{#1}{\mbox{\ensuremath{\displaystyle #2}}}}

\newcommand{\Hball}{{H_{\mathrm{ball}}}}
\newcommand{\Hballhalf}{{H_{\mathrm{ball}/2}}}
\newcommand{\Lball}{\L_{\mathrm{ball}}}
\newcommand{\Hsim}{H'_{\mathrm{sim}}}
\newcommand{\Hlog}{H'_{\mathrm{logical}}}
\usetikzlibrary{patterns,patterns.meta,decorations.pathmorphing}
\definecolor{colExclude}{gray}{0.65}
\colorlet{colUnitary}{red!90!black}
\tikzset{
  x=7.5mm,
  y=7.5mm,
  transition/.style={thick,draw},
  exclude/.style={pattern=north east lines, pattern color=colExclude},
  unitary/.style={transition,color=colUnitary,-latex},
  connect/.style={transition,color=blue,thick,decorate,decoration={zigzag,segment length=3mm,amplitude=0.7mm}},
  conditional/.style={transition,dashed},
  arrow cont/.style={->,very thick,dotted},
  arrow cont unitary/.style={->,very thick,dotted,color=colUnitary},
  component/.style={color=black!20,fill=black!4,thick}
}
\newcommand{\baseGridW}{5}
\newcommand{\baseGridH}{4}
\newcommand{\dotSize}{0.085}
\newcommand{\makeGrid}[2]{
  \foreach \i in {1,...,#1} {
    \foreach \j in {1,...,#2} {
        \fill (\i,\j) circle (\dotSize);
    }
}}
\newcommand{\makeGridBR}[2]{
  \foreach \i in {2,...,#1} {
    \fill (\i-1,1) circle (\dotSize);
  }
  \foreach \i in {1,...,#1} {
    \foreach \j in {2,...,#2} {
        \fill (\i,\j) circle (\dotSize);
    }
}}
\newcommand{\baseGrid}{
  \makeGrid{\baseGridW}{\baseGridH}
  \useasboundingbox (0,0.5) rectangle (\baseGridW+1,\baseGridH+0.5);
}
\newcommand\dC{8}

\newcommand{\makeEdgesX}[4]{\foreach \y in {#2,...,#3} {
	\draw[#4] (#1,\dC-\y) -- (#1+1,\dC-\y);}}
\newcommand{\makeEdgesY}[4]{\foreach \y in {#2,...,#3} {
	\draw[#4] (\y,\dC-#1) -- (\y,\dC-#1-1);}}
\newcommand{\makeEdgesXY}[4]{
	\makeEdgesX{#1}{#2}{#3}{#4}
	\makeEdgesY{#1}{#2}{#3}{#4}}

\newcommand{\makeArrowUpX}[3]{
	\path[#3] (#1+0.5,\dC-#2+0.25) edge (#1+0.5,\dC-#2+1);}
\newcommand{\makeArrowUpY}[3]{
	\path[#3] (#2-0.25,\dC-#1-0.5) edge (#2-1,\dC-#1-0.5);}
\newcommand{\makeArrowUpXY}[3]{
  \makeArrowUpX{#1}{#2}{#3}
  \makeArrowUpY{#1}{#2}{#3}}

\newcommand{\makeArrowDownX}[3]{
	\path[#3] (#1+0.5,\dC-#2-0.25) edge (#1+0.5,\dC-#2-1);}
\newcommand{\makeArrowDownY}[3]{
	\path[#3] (#2+0.25,\dC-#1-0.5) edge (#2+1,\dC-#1-0.5);}
\newcommand{\makeArrowDownXY}[3]{
	\makeArrowDownX{#1}{#2}{#3}
	\makeArrowDownY{#1}{#2}{#3}}

\newcommand{\makeEdgesUpXY}[5]{
	\makeEdgesXY{#1}{#2}{#3}{#4}
	\makeArrowUpXY{#1}{#2}{#5}}
\newcommand{\makeEdgesDownXY}[5]{
	\makeEdgesXY{#1}{#2}{#3}{#4}
	\makeArrowDownXY{#1}{#3}{#5}
}
\newcommand{\makeEdgesUpDownXY}[5]{
	\makeEdgesXY{#1}{#2}{#3}{#4}
	\makeArrowUpXY{#1}{#2}{#5}
	\makeArrowDownXY{#1}{#3}{#5}}

\newcommand\TL{\mathcal{T\mkern-3muL}}
\newcommand\TR{\mathcal{TR}}
\newcommand\BL{\mathcal{B\mkern-1muL}}
\newcommand\BR{\mathcal{BR}}
\newcommand\M{\mathcal{M}}

\def\figcommon{}

\title{Quantum $2$-SAT on low dimensional systems is $\QMAo$-complete: Direct embeddings and black-box simulation}
\author{Dorian Rudolph\footnotemark[1] \and Sevag Gharibian\footnote{Department of Computer Science and Institute for Photonic Quantum Systems (PhoQS), Paderborn University, Germany. Email: \{sevag.gharibian, dorian.rudolph\}@upb.de.} \and Daniel Nagaj\footnote{Institute of Physics, Slovak Academy of Sciences. Email: dnagaj@gmail.com}}

\begin{document}

\maketitle

\begin{abstract}
  Despite the fundamental role the Quantum Satisfiability (QSAT) problem has played in quantum complexity theory, a central question remains open: At which local dimension does the complexity of QSAT transition from ``easy'' to ``hard''?
  Here, we study QSAT with each constraint acting on a $k$-dimensional and $l$-dimensional qudit pair, denoted $(k,l)$-QSAT.
  Our first main result shows that, surprisingly, QSAT on qubits can remain $\QMAo$-hard, in that $(2,5)$-QSAT is $\QMAo$-complete.
  In contrast, $(2,2)$-QSAT (i.e. Quantum $2$-SAT on qu\emph{b}its) is well-known to be poly-time solvable [Bravyi, 2006].
  Our second main result proves that $(3,d)$-QSAT on the 1D line with $d\in O(1)$ is also $\QMAo$-hard.
  Finally, we initiate the study of $(2,d)$-QSAT on the 1D line by giving a frustration-free 1D Hamiltonian with a unique, entangled ground state.

  As implied by our title, our first result uses a \emph{direct embedding}: We combine a novel clock construction with the 2D circuit-to-Hamiltonian construction of [Gosset and Nagaj, 2013].
  Of note is a new simplified and \emph{analytic} proof for the latter (as opposed to a partially numeric proof in [GN13]). This exploits Unitary Labelled Graphs [Bausch, Cubitt, Ozols, 2017] together with a new ``Nullspace Connection Lemma'', allowing us to break low energy analyses into small patches of projectors, and to improve the soundness analysis of [GN13] from $\Omega(1/T^6)$ to $\Omega(1/T^2)$, for $T$ the number of gates.
  Our second result goes via \emph{black-box} reduction: Given an \emph{arbitrary} 1D Hamiltonian $H$ on $d'$-dimensional qudits, we show how to embed it into an effective 1D $(3,d)$-QSAT instance, for $d\in O(1)$.
  Our approach may be viewed as a weaker notion of ``simulation'' (\emph{\`{a} la} [Bravyi, Hastings 2017], [Cubitt, Montanaro, Piddock 2018]). As far as we are aware, this gives the first ``black-box simulation''-based $\QMAo$-hardness result.
\end{abstract}

\section{Introduction}

Boolean satisfiability problems have long served as a testbed for probing the boundary between ``easy'' (i.e. poly-time solvable) versus ``hard'' (e.g. \NP-complete) computational problems.
A striking early example of this is the fact that while $3$-SAT is NP-complete~\cite{cookComplexityTheoremprovingProcedures1971,l.levinUniversalSearchProblems1973,karpReducibilityCombinatorialProblems1972}, $2$-SAT is in P~\cite{w.v.quineCoresPrimeImplicants1959,m.davisComputingProcedureQuantification1960,m.r.kromDecisionProblemClass1967,s.evenComplexityTimeTable1976,b.aspvallLineartimeAlgorithmTesting1979, c.papadimitriouSelectingSatisfyingTruth1991}.
Despite this, \emph{MAX}-$2$-SAT (i.e. what is the \emph{maximum} number of satisfiable clauses of a $2$-CNF formula?) remains NP-complete ~\cite{gareySimplifiedNPcompleteGraph1976}!
Thus, the border between tractable and intractable can often be intricate, hiding abrupt transitions in complexity.

In the quantum setting, generalizations of $k$-SAT and MAX-$k$-SAT have similarly played a central role, additionally due to their strong physical motivation. The input here is a \emph{$k$-local Hamiltonian} $H=\sum_iH_i$ acting on $n$ qubits, which is a $2^n\times 2^n$ complex Hermitian matrix (a quantum analogue of a Boolean formula on $n$ bits), given via a succinct description $\set{H_i}$. Here, each $H_i$ is a $2^k\times 2^k$ operator acting on\footnote{Formally, if $H_i$ acts on a subset $S\subseteq[n]$ of qubits, to make the dimensions match we consider $(H_i)_S\otimes I_{[n]\setminus S}$.} $k$ out of $n$ qubits (i.e. a local quantum clause). Given $H$, the goal is to compute the smallest eigenvalue $\lmin(H)$ of $H$, known as the \emph{ground state energy}. The corresponding eigenvector, in turn, is the \emph{ground state}. This $k$-local Hamiltonian problem ($k$-LH) generalizes MAX-$k$-SAT, and formalizes the question: If a many-body quantum system is cooled to near absolute zero, what energy level will the system relax into? The complexity of $k$-LH is well-understood, and analogous to MAX-$2$-SAT, even $2$-LH is complete for\footnote{QMA is the bounded-error quantum analogue of NP, now with poly-size quantum proof and quantum verifier~\cite{KSV02}.} Quantum Merlin-Arthur (QMA)~\cite{KSV02, j.kempeComplexityLocalHamiltonian2006,cubittComplexityClassificationLocal2016a}.

The quantum generalization of $k$-SAT (as opposed to \emph{MAX}-$k$-SAT), on the other hand, is generally \emph{less} understood. In contrast to $k$-LH, one now asks whether there exists a ground state of $H$ which is simultaneously a ground state for \emph{each local term} $H_i$ (analogous to a string satisfying \emph{all} clauses of a $k$-SAT formula).
Formally, in Quantum $k$-SAT ($k$-QSAT), each $H_i$ is now a projector, and one asks whether $\lmin(H)=0$.
As in the classical setting, it is known that the locality $k$ leads to a complexity transition: On the one hand, Gosset and Nagaj~\cite{GN13} proved that $3$-QSAT is\footnote{$\QMAo$ is \QMA but with perfect completeness. Note that while $\MA=\MAo$~\cite{zachosProbabilisticQuantifiersVs1987}, whether $\QMAo=\QMA$ remains a major open question.} $\QMAo$-complete, while on the other hand, Bravyi gave~\cite{bravyiEfficientAlgorithmQuantum2006} a poly-time classical algorithm for $2$-QSAT (in fact, $2$-QSAT is solvable in linear time~\cite{aradLinearTimeAlgorithm2016,beaudrapLinearTimeAlgorithm2016}). \\
\vspace{-1mm}

\noindent \emph{Systems of higher local dimension.} In the quantum setting, however, there is an additional, physically motivated direction to probe for complexity transitions for $2$-QSAT --- \emph{systems of higher local dimension}.
Perhaps the most striking example of this is that, while Boolean satisfiability problems in 1D are efficiently solvable via dynamic programming (even for $d$-level systems instead of bits), Aharonov, Gottesman, Irani and Kempe showed~\cite{aharonovPowerQuantumSystems2009} that $2$-LH on the line remains \QMA-complete, \emph{so long as} one uses local dimension $d=12$! An analogous result for $2$-QSAT with $d=11$ was subsequently given by Nagaj~\cite{nagajLocalHamiltoniansQuantum2008}. This raises the guiding question of this work:
\begin{center}
    \emph{What is the smallest local dimension that can encode a \QMAo-hard problem?}
\end{center}
There are three results we are aware of here. Define $(k,l)$-QSAT as $2$-QSAT where each constraint acts on a qu-$k$-it and a qu-$l$-it, i.e. on $\CC^k\otimes \CC^l$. (When $k\neq l$, for this to be well-defined, the interaction graph of the instance must be bipartite.) Chen, Chen, Duan, Ji, and Zheng showed~\cite{chenNogoTheoremOneway2011} that $2$-QSAT on qubits, $(2,2)$-QSAT, cannot encode $\QMAo$-hard problems unless $\NP=\QMAo$, as the ground space always contains an NP witness.
On the other hand, Bravyi, Caha, Movassagh, Nagaj, and Shor gave a frustration-free\footnote{A \emph{frustration-free} Hamiltonian is a YES instance of QSAT, i.e. a local Hamiltonian $H\succeq 0$ with $\lmin(H)=0$.} qutrit construction (i.e. on local dimension $d=3$) on the 1D line\footnote{``On the line'' means $H=\sum_{i=1}^m H_{i,i+1}$, i.e. the qudits can be depicted in a sequence with each consecutive constraint acting on the next pair of qudits in the sequence.}~\cite{bravyiCriticalityFrustrationQuantum2012} with a \emph{unique}, \emph{entangled} ground state. While this construction does not encode a computation (and thus does not give $\QMAo$-hardness), it is an important first step in that it shows even such low dimensional systems can encode entangled witnesses (entanglement is necessary, otherwise an NP witness is possible).
Together, these works~\cite{chenNogoTheoremOneway2011} and~\cite{bravyiCriticalityFrustrationQuantum2012} suggest that qu\emph{b}it systems are a no-go barrier for $\QMAo$-hardness. Prior to these, Eldar and Regev~\cite{ER08} came closest to establishing a result about qubit systems, showing that $(3,5)$-QSAT is $\QMAo$-hard (on general graphs).\footnote{$\QMA_1$-hardness of $(4,9)$-QSAT is claimed without proof in \cite{NM07}.}
But the key question remains open --- \emph{Can qubit systems support $\QMAo$-hardness for Quantum $2$-SAT, i.e. is $(2,k)$-QSAT $\QMAo$-hard for some $k\in O(1)$?}

\paragraph{Our results.} We show two main results, along with a third preliminary one. \\

\newcommand\QMAoH{{\textbf{QMA$_1$-C}}}
\begin{table}[h]
\begin{center}
  \begin{tabularx}{.59\linewidth}{l||l|l|l|l|l}
  &$2$&$3$&$4$&$5$&$6$\\\hline\hline
  $2$&$\in\p$&$\NP$-H&$\NP$-H&\QMAoH&\QMAoH\\
  $3$&&$\NP$-H&\QMAoH&$\QMAo$-C&$\QMAo$-C\\
  $4$&&&\QMAoH&$\QMAo$-C&$\QMAo$-C\\
  $5$&&&&$\QMAo$-C&$\QMAo$-C
  \end{tabularx}
\end{center}
\caption{Summary of the known complexity results of$(k,l)\hQSAT$. New results from this work are bold. `-H'/`-C' denote hardness/completeness.}
\label{table:results}
\end{table}

\noindent\emph{1. $\QMAo$-hardness for qubit systems.} The complexity of $(k,l)\hQSAT$ including our new results is summarized in \Cref{table:results}.
The first main result is as follows.

\begin{restatable}{theorem}{thmTwoFive}\label{thm:25}
  $(2,5)$-QSAT is $\QMAo$-complete with soundness $\Omega(1/T^2)$.
\end{restatable}

\noindent Thus, qubit systems \emph{can} encode $\QMAo$-hardness.
Let us be clear that the surprising part of this is \emph{not} that this setting is intractable --- indeed, even classical $(2,3)$-SAT is known to be NP-hard~(e.g. \cite{nagajLocalHamiltoniansQuantum2008}).
What \emph{is} surprising is that one can encode \emph{quantum} verifications in such low dimension, for two reasons.
First, \emph{a priori} a $2$-dimensional space for $2$-local constraints appears too limited to \emph{exactly}\footnote{An \emph{exact} encoding appears needed by definition of \QSAT. This is in strong contrast to $2$-LH, where \emph{approximate} encodings are allowed (since all constraints need not be simultaneously satisfiable) via perturbation theory~\cite{j.kempeComplexityLocalHamiltonian2006,bravyiSchriefferWolffTransformation2011}.} encode a computation --- a two-dimensional space only appears to suffice to encode a ``data qubit'', so where does one encode the ``clock'' tracking the computation?
Second, the entanglement of $2\times d$ systems is generally easier to characterize than that of $d\times d$ systems.
For example, whether a $2\times 2$ or $2\times 3$ system is entangled is detectable via Peres' Positive Partial Transpose (PPT) criterion~\cite{peresSeparabilityCriterionDensity1996}, whereas detecting entanglement for $d\times d$ systems (for polynomial $d$) becomes strongly NP-hard~\cite{l.gurvitsClassicalDeterministicComplexity2003,l.ioannouComputationalComplexityQuantum2007,s.gharibianStrongNPhardnessQuantum2010}.
And recall that a ``sufficiently entangled'' ground space is necessary to encode $\QMAo$-hardness.

As a complementary result, we show that one can ``trade'' dimensions in the construction above if one is careful, i.e. the $5$ in $(2,5)$ can be reduced to $4$ at the expense of increasing $2$ to $3$.
\begin{restatable}{theorem}{thmThreeFour}\label{thm:34}
  $(3,4)$-QSAT is $\QMAo$-complete with soundness $\Omega(1/T^2)$.
\end{restatable}

We remark that obtaining Theorems \ref{thm:25} and \ref{thm:34} is significantly more involved than the previous best $\QMAo$-hardness for $(3,5)$-QSAT~\cite{ER08} (details in ``Techniques'' below), and in particular requires (among other ideas) a simplified anlaysis of the advanced $2D$-Hamiltonian framework of Gosset and Nagaj for $(2,2,2)$-QSAT (i.e. $3$-QSAT on qubits)~\cite{GN13}, a new clock construction, and the Unitary Labelled Graphs of Bausch, Cubitt and Ozols~\cite{BCO17}.
One of the payoffs is that we obtain a ``tight''\footnote{By ``tight'', we mean that it is not currently known for either QSAT or LH how to get a promise gap larger than $\Omega(1/T^2)$~\cite{Wat19}.} soundness gap of $\Omega(1/T^2)$ for \Cref{thm:25}, compared to the $\Omega(1/T^6)$ gap of~\cite{GN13}.
This allows us to recover the $3$-QSAT hardness results of~\cite{GN13}, but with improved soundness:

\begin{restatable}{theorem}{thmThree}\label{thm:3}
  $3$-QSAT is $\QMAo$-complete with soundness $\Omega(1/T^2)$.
\end{restatable}

\noindent\emph{2. Low dimensional systems on the line.}  Our next main result is the following.

\begin{restatable}{theorem}{thmThreeD}\label{thm:3d}
  $(3,d)$-QSAT on a line is $\QMAo$-complete with $d=O(1)$.
\end{restatable}
\noindent This improves significantly on the previous best $(11,11)$-QSAT 1D $\QMAo$-hardness construction of Nagaj~\cite{nagajLocalHamiltoniansQuantum2008}, showing that frustration free Hamiltonians on qutrits can encode not just entangled ground states (\emph{cf.}~\cite{bravyiCriticalityFrustrationQuantum2012}), but also $\QMAo$-hard computations.

For clarity, there is a trade-off in the parameter, $d$, which we now elaborate. 
A key novelty of \Cref{thm:3d} is its proof via \emph{black-box} reduction (see Techniques below), as opposed to a direct embedding of a $\QMAo$ computation.
Specifically, given an \emph{arbitrary} 1D Hamiltonian $H$ on $d'$-dimensional qudits, we embed it into an effective 1D $(3,d)$-QSAT instance with $d\in O(1)$.
On the not-so-positive side, the generality of this approach means that an input 1D Hamiltonian $H$ on $d$-dimensional qudits is mapped to a 1D Hamiltonian on constraints of dimension $(3,O((d')^4))$, i.e. the first dimension drops to $3$ at the expense of the second dimension increasing.
Thus, for example, if we plug in the current best known 1D $\QMAo$-hardness construction~\cite{nagajLocalHamiltoniansQuantum2008}, \Cref{thm:3d} gives $\QMAo$-hardness for $(3,76176)$-QSAT.

On the positive side, however, our technique is the first use of (a weaker notion) of the influential idea of \emph{local simulation} (Bravyi and Hastings~\cite{BH17}, Cubitt, Montanaro, and Piddock~\cite{cubittUniversalQuantumHamiltonians2018}; see Definition 4.1 of \cite{gharibianOracleComplexityClasses2020} for a simpler statement) in the study of QSAT.
Roughly, in such simulations, given a local Hamiltonian $H$ on $n$ qubits, one typically a local isometry $V$ to each qubit, i.e. maps $H\mapsto V^{\otimes n} H (V^\dagger)^{\otimes n}$, blowing up the input space $A$ into a larger, ``logical'' space $B$.
By cleverly choosing an appropriate Hamiltonian $H'$ on $B$, one forces the low-energy space of $H'$ to approximate $H$.
Traditionally, the drawback of this approach is its reliance on perturbation theory, which necessarily gives rise to\footnote{In words, perturbation theory is only known to be able to show $\QMA$-hardness for $k$-LH, as opposed to $\QMAo$-hardness for $k$-QSAT.} \emph{frustrated} Hamiltonians $H'$.
Here, however, we show for the first time that a weaker\footnote{For clarity, the simulations of \cite{BH17,cubittUniversalQuantumHamiltonians2018} reproduce the whole target Hamiltonian $H$, whereas our approach is weaker in that we prove simulation of only $H$'s null space.} version of such local embeddings can be designed even for the frustration-free setting, ultimately yielding \Cref{thm:3d}.\\

\noindent\emph{3. Towards qubits on the line.} Finally, we initiate the study $(2,d)$ on a line by proving that even a frustration-free Hamiltonian on a line of alternating particles with dimensions $2$ and $4$ can have a unique entangled ground state.
\begin{restatable}{theorem}{thmTwoFourLine}\label{thm:24-line}
  Consider a line of particles $2n$ particles such that the $i$-th particle has dimension $2$ for even $i$ and $4$ for odd $i$.
  There exists a Hamiltonian $H = \sum_{i=1}^{n} A_{2i-1,2i} + \sum_{i=1}^{n-1}B_{2i,2i+1} + L_{1,2} + R_{2n-1,2n}$, where $A,B,L,R$ are $2$-local projectors, such that $\Null(H) = \Span\{\ket\psi\}$ and $\ket{\psi}$ is entangled across all cuts.\footnote{The Schmidt rank is $\Theta(1)$, but we do not explicitly analyze it here.}
\end{restatable}

\paragraph{Techniques.} We focus on our main results, \Cref{thm:25} and \Cref{thm:3d}.\\ 

\noindent\emph{$\QMAo$-hardness for $(2,5)$-QSAT.} We begin with \Cref{thm:25}, which is proven via a direct embedding of a $\QMAo$ computation into a $(2,d)$-QSAT instance. 
Our first challenge is to break the ``qutrit barrier'' by embedding a clock into the nullspace of a $(2,d)$-Hamiltonian.
When constructing a clock for a circuit Hamiltonian, one has to ensure that each transition from timestep $i$ to $i+1$ has support precisely on these timesteps.
However, with only qubit-systems at our disposal to act as ``auxiliary particles'' (as opposed to qutrits in \cite{ER08}), even the construction of a clock (never mind embedding a computation!) is not obvious.

Our starting observation is that a simple sequence of clock states $\ket{100},\ket{200},\ket{210},\ket{220},\ket{221},\dots$ can be implemented in the nullspace of a $(3,3)$-Hamiltonian.
Note that in this clock, only a single qutrit changes in each step.
Hence, our key idea is to implement ``logical qutrits'' by adding ``indicator qubits'' to $6$-dimensional qudits, so that the logical qutrit $\ket{x}$ becomes 
\begin{equation}
    \ket{x}\mapsto \ket{x}_\alpha\ket{000}_\beta+\ket{x'}_\alpha\ket{0^{x}10^{2-x}}_\beta\quad\text{for }x\in\{0,1,2\}, x'=x+3\in\set{3,4,5}.
\end{equation}
Here, $\alpha$ is $6$-dimensional, and $\beta$ consists of three qubits labelled $\beta_0,\beta_1,\beta_2$. Now, a $\ket{10}\leftrightarrow\ket{20}$ transition on the original $(3,3)$ space can be realized with a $(2,6)$-projector onto $\ket{11}_{\alpha,\beta_0}-\ket{21}_{\alpha,\beta_0}$.
This basic principle generalizes to qudits of any dimension, and combined with the $(3,5)$-Hamiltonian of \cite{ER08}, gives $\QMAo$-hardness of $(2,10)$-QSAT.
We can remove some indicator qubits to improve this to $(2,7)$-QSAT (\Cref{sec:intuition}) as the indicator requires an extra qudit dimension ($\ket{x}$ and $\ket{x'}$).

Further improvement to $(2,5)$-QSAT requires much more work.
For that, we employ the $2D$-clock \cite{GN13}, which geometrically realizes the idea of Eldar and Regev~\cite{ER08} of implementing CNOT gates by ``routing'' states with a $\ket{0}$ and $\ket{1}$ in the control register along different paths.
The $2D$-Hamiltonian thus has more relaxed requirements for the clock, so that a $(3,4)$-system suffices.
The key difference is that \cite{ER08} requires $3$ ``alive states'' to implement CNOT gates with a ``triangle gadget'' (the triangle contains the two paths), whereas \cite{GN13} only needs $2$ ``alive states'' by using the second dimension to realize two distinct paths.
Then we simulate this $(3,4)$-clock on a $(2,5)$-system via the indicator qubit principle.
Note that when implementing a logical $4$-dit on a $(2,5)$-system, we are only allowed to use a single indicator qubit.
Thus, we need a very carefully crafted $(3,4)$-clock, and further ``technical tricks''.

Finally, to analyze the $2D$-Hamiltonian, we prove a novel technical lemma, which we dub the ``Nullspace Connection Lemma'', because it enables us to split the $2D$-Hamiltonian into smaller gadgets, each of which implements a small part of the history state\footnote{A \emph{history state} is a quantum analogue of a tableau from the Cook-Levin theorem.}.
The gadgets are then connected with additional transitions to form the complete Hamiltonian, whose nullspace is spanned by superpositions of the gadget's history states.
This lemma can also be applied to, e.g., the original circuit Hamiltonian of Kitaev \cite{KSV02} (see \Cref{rem:nullspace-kitaev}), matching the $\Omega(1/T^2)$ soundness (smallest non-zero eigenvalue) therein.
Furthermore, the Connection Lemma is proven directly via the Geometric Lemma~\cite{KSV02} and does not require transformation to the Laplacian matrix of a random walk, unlike~\cite{KSV02}.
Our main application of the Connection Lemma is to give a simplified proof for the $2D$-Hamiltonian with improved soundness.
Since the Connection Lemma requires modifications to the Hamiltonian of \cite{GN13}, we present our own variant, which is slightly more compact (using $6M+1$ clock states as opposed to $9M+3$ for $M$ gates).
Finally, we do not rely on numerical methods to derive the nullspaces of the individual gadgets (\emph{cf.} the gadget analysis of \cite{GN13}, which required numerics), and instead give a formal proof based on the theory of \emph{unitary labeled graphs} \cite{BCO17}.
The upshot is that our overall approach is very flexible --- by combining unitary labelled graphs with the Connection Lemma, one can in principle analyze combinations of Hamiltonian gadgets beyond just our 2D setting with relative ease.\\

\noindent\emph{$\QMAo$-hardness for $(3,d)$-QSAT on the line.} As mentioned earlier, in contrast to the direct embedding for \Cref{thm:25}, for \Cref{thm:3d} we instead use a black-box simulation.  
In other words, we do not give a new circuit-to-Hamiltonian mapping, but instead bootstrap the prior $\QMAo$-completeness result of QSAT on a line of qu$d$its with $d=11$ due to Nagaj~\cite{nagajLocalHamiltoniansQuantum2008}.
We treat that Hamiltonian as a black box and construct an embedding of a general $1D$-Hamiltonian $H$ on a line of qudits into an Hamiltonian $H'$ on an alternating line of qu$d'$its and qutrits.
Each qu$d'$it is treated as two logical qu$d''$its with $d''=\sqrt{d'}$ (see \Cref{fig:H'}).
We construct a Hamiltonian $\Hlog$ that restricts the $(d'',3,d'')$ systems to a $d$-dimensional subspace, which acts as a logical qu$d$it.
This logical subspace has the key feature that, in a sense, the logical qu$d$it can be ``accessed'' from either its left or right qu$d''$it.
This allows us to encode the $2$-local terms of $H$ as $1$-local terms acting on the qu$d'$its of $H'$.
As previously described, our embedding can be seen as a weaker notion of \emph{simulation} in the sense of \cite{BH17,cubittUniversalQuantumHamiltonians2018}, in that formally our embedding is achieved via application of local isometries, followed by additional constraints on the logical space (Equations (\ref{eqn:iso1})-(\ref{eqn:iso3})).

\paragraph{Open questions.} Although we have shown that qubit systems can support $\QMAo$-hard problems, the frontier for characterizing the complexity transition of local Hamiltonian problems from low to high local dimension remains challenging.
In our setting, in particular, the main open question is whether one can obtain $\QMAo$-hardness even for $(2,3)$-QSAT? This would complete the complexity characterization for $(k,l)$-QSAT, as recall $(2,2)$-QSAT (i.e. $2$-QSAT) is in P~\cite{bravyiEfficientAlgorithmQuantum2006}.
Getting this down to $(2,3)$-QSAT (or even $(3,3)$-QSAT), however, appears difficult, requiring ideas beyond those introduced here.

As for the 1D line, the best hardness results for $2$-LH and $2$-QSAT are on $8$-dimensional~\cite{HNN13} and $11$-dimensional~\cite{nagajLocalHamiltoniansQuantum2008} systems, respectively.
Is 1D $2$-QSAT on qudits with $2< d<11$ $\QMAo$-hard?
We showed that 1D $(3,d)$-QSAT is $\QMAo$-hard, but due to our black-box approach we only get $d=76176$.
So it seems likely that this $d$ can be improved significantly.
Perhaps more interesting is the question whether 1D $(2,d)$-QSAT is still $\QMAo$-hard.
We showed that even 1D $(2,4)$-dimensional constraints can support a unique globally entangled ground state (\Cref{thm:24-line}), but this construction alone does not embed a computation, and thus does not yield $\QMAo$-hardness.
Can this be bootstrapped to obtain $\QMAo$-hardness for 1D $(2,d)$-QSAT?
For 1D $2$-LH, the situation is even worse --- on qubits, these systems can only be efficiently solved in the presence of a constant spectral gap~\cite{landauPolynomialTimeAlgorithm2015}.
In contrast, for inverse polynomial gap, and even with the promise of an NP witness (i.e. via Matrix Product State), the problem is NP-complete~\cite{schuchComputationalDifficultyFinding2008}. What is the complexity of 1D $2$-LH on qudits with $2< d<8$?

\paragraph{Organization.}

\Cref{sec:prelim} gives formal definitions for $\QMA_1$, $(k,l)$-QSAT, and states the Geometric Lemma with various corollaries.
\Cref{sec:2D-Hamiltonian} describes the 2D-Hamiltonian and proves its soundness.
\Cref{sec:clock} constructs the $(2,5)$-clock and proves $\QMA_1$-completeness of $(2,5)$- and $(3,4)$-QSAT.
\Cref{sec:connection-lemma} proves the Nullspace Connection Lemma.
\Cref{sec:3d} proves the $\QMAo$-completeness of $(3,d)$-QSAT on a line.
\Cref{sec:line} gives the construction of the $(2,4)$-Hamiltonian on a line with entangled ground space.

\section{Preliminaries}\label{sec:prelim}

In this section, we formally introduce $\QMAo$ and elaborate on the Geometric Lemma.

\subsection{\texorpdfstring{$\QMAo$}{PDFstring}}

\newcommand{\gateset}{\mathcal{G}}
The complexity class $\QMA_1$ is defined in the same way as $\QMA$, but with the additional requirement of \emph{perfect completeness}, i.e., in the YES-case, there exists a proof that the verifier accepts with a probability of exactly $1$.
Consequently, $\QMA_1$ is not known to be independent of the gate set~\cite{GN13}, as approximate decompositions of arbitrary unitaries generally breaks perfect completeness.
Therefore, we have to fix a gate set before we define $\QMAo$, and here we follow~\cite{GN13} in choosing the ``Clifford + T'' gate set $\gateset = \{\widehat{H}, T, \CNOT\}$, where $\widehat H$ denotes the Hadamard gate.
Giles and Selinger~\cite{GS13} have proven that a unitary can be synthesized exactly with gate set $\gateset$ iff its entries are in the ring $\ZZ[\frac{1}{\sqrt2},i]$.

\begin{definition}[$\QMAo$]\label{def:QMA1}
  A promise problem $A=(\ayes,\ano)$ is in $\QMAo$ if there exists a poly-time uniform family of quantum circuits $\set{Q_x}$ with gate set $\gateset$ such that:
  \begin{itemize}
    \item (Completeness) If $x\in\ayes$, then there exists a proof $\ket{\psi}$ with $\Pr[Q_x\text{ accepts }\ket{\psi}] = 1$.
    \item (Soundness) If $x\in\ano$, then then for all proofs $\ket{\psi}$, $\Pr[Q_x\text{ accepts }\ket{\psi}] \le 1-1/\poly(\abs{x})$.
  \end{itemize}
\end{definition}

\begin{definition}[$(k,l)$-QSAT]\label{def:klQSAT}
  Consider a system of $k$-dimensional and $l$-dimensional particles, denoted $k_i,i\in [n_k]$ and $l_j,j\in [n_l]$, respectively.
  In the $(k,l)$-QSAT problem, the input is a $(k,l)$-Hamiltonian $H=\sum_{i\in[n_k],j\in[n_l]} \Pi_{k_i,l_j}$ with $2$-local projectors $\Pi_{k_i,l_j}$ acting non-trivially only on particles $k_i$ and $l_j$.
  Decide:
  \begin{itemize}
    \item (YES) $\lmin(H)=0$.
    \item (NO) $\lmin(H)\ge 1/\poly(n_l+n_k)$.
  \end{itemize}
\end{definition}

Note that the projectors of $(k,l)$-QSAT need to have a specific form such that the problem is contained in $\QMA_1$ with our chosen gate set (see \Cref{sec:QMA1-completeness}).

\paragraph{Unitary Labelled Graphs.} Our analysis of the 2D Hamiltonian in Section~\Cref{sec:2D-Hamiltonian} utilizes the notion of Unitary Labelled Graphs of Bausch, Cubitt and Ozols~\cite{BCO17}, we which now define. The power of the ULG framework is that it simplifies  the characterization of the null space of Kitaev's propagation Hamiltonian~\cite{KSV02} when time steps do not necessarily occur in a simple linear progression (specifically, for propagation Hamiltonians defined on ``simple'' ULGs).

\begin{definition}[Unitary Labelled Graph (ULG)~\cite{BCO17}]\label{def:ULG}
    Given an undirected graph $G=(V,E)$ with no self-loops, a \emph{unitary labelled graph} is a triple $(G,\set{\spa{H}_v}_{v\in V},g)$ such that:
    \begin{itemize}
        \item each $\spa{H}_v$ denotes a Hilbert space,
        \item each edge $(u,v)\in E$ is labelled by a unitary $g(u,v)=:U$, so that $(v,u)\in E$ (the same edge in the reverse direction) has label $g(v,u)=U^\dagger$.
    \end{itemize}
\end{definition}

\begin{definition}[Simple ULG~\cite{BCO17}]\label{def:simple}
  For ULG  $(G,\set{\spa{H}_v}_{v\in V},g)$, if for all $u,v\in V$, the product of unitaries along any directed path connecting $a$ to $b$ is equal, then ULG is \emph{simple}. (Equivalently, the product of unitaries along any directed cycle is $I$.)
\end{definition}

\subsection{Geometric Lemma}

In our proofs, we frequently apply Kitaev's Geometric Lemma~\cite{KSV02} as well as its extension to the frustration-free case due to Gosset and Nagaj~\cite{GN13}, where we are interested in lower-bounding the smallest non-zero eigenvalue of a Hamiltonian $H$, denoted $\gamma(H)$.
In the following, we also give further refinements of these statements.
As in \cite{GN13}, we use the notation $H|_S = \Pi_S H \Pi_S$ for the restriction of the Hamiltonian $H$ to the subspace $S$, where $\Pi_S$ is the projector onto $S$.%
\footnote{Note, this is not the standard restriction of linear map to a subspace since $H$ does not necessarily map $S$ to itself.}

\begin{lemma}[Kitaev's Geometric Lemma \cite{KSV02} as stated in \cite{GN13}]\label{lem:geometric}
  Let $H=H_A + H_B$ with $H_A\succeq0$ and $H_B\succeq0$.
  Let $S = \Null(H_A)$ and $\Pi_B$ be the projector onto $\Null(H_B)$.
  Suppose $\Null(H) = \{0\}$. Then
  \begin{equation}
    \gamma(H) \ge \min\{\gamma(H_A),\gamma(H_B)\}\cdot (1-\sqrt{c}),
  \end{equation}
  where
  \begin{equation}
    c = \max_{\ket{v}\in S:\braket{v}{v}=1}\bra v\Pi_B\ket v.
  \end{equation}
\end{lemma}
\begin{corollary}[{\cite[Corollary 1]{GN13}}]\label{cor:GN:cor1}
  Let $H = H_A + H_B$ where $H_A\succeq0$ and $H_B\succeq0$ each have nonempty nullspaces.
  Let $\Gamma$ be the subspace of states in $\Null(H_A)$ that are orthogonal to $\Null(H)$, and let $\Pi_B$ be the projector onto $\Null(H_B)$. Then
  \begin{equation}
    \gamma(H)\ge\min\{\gamma(H_A),\gamma(H_B)\}\cdot\left(1-\sqrt d\right),
  \end{equation}
  where
  \begin{equation}
    d = \norm{\Pi_B|_\Gamma} = \max_{\ket v\in\Gamma:\braket vv=1}\bra v\Pi_B\ket v.
  \end{equation}
\end{corollary}
We give a slightly tighter statement of \cite[Corollary 2]{GN13}\footnote{The only difference is that we have $\norm{H_B\vert_S}$ instead of $\norm{H_B}$ in the denominator.}:
\begin{corollary}\label{cor:GN:cor2}
  Let $H = H_A + H_B$ where $H_A\succeq0$ and $H_B\succeq0$ each have nonempty nullspaces.
  Let $S = \Null(H_A)$ and suppose $H_B|_S$ is not the zero matrix. Then
  \begin{equation}
    \gamma(H)\ge\min\{\gamma(H_A),\gamma(H_B)\}\cdot\frac{\gamma(H_B|_S)}{2\norm{H_B\vert_S}}.
  \end{equation}
\end{corollary}
\begin{proof}
  As stated in \cite{GN13}, $\gamma(H_B\vert_S) = \min_{v\in\Gamma:\braket vv =1}\braketb{v}{H_B}$ with $\Gamma$ as defined in \Cref{cor:GN:cor1}.
  Hence for all unit $\ket v \in\Gamma\subseteq S$,
  \begin{equation}
    \gamma(H_B\vert_S)\le \braketb{v}{H_B}\le \braketb{v}{(\I-\Pi_B)}\norm{H_B\vert_S}
  \end{equation}
  and thus $d\le 1-\gamma(H_B\vert_S)/\norm{H_B\vert_S}$.
  The statement then follows from $1-\sqrt{1-x}\ge \frac x2$ for $x\in[0,1]$.
\end{proof}
\begin{corollary}\label{cor:geometric3}
  Let $H = H_A + H_B$ where $H_A\succeq0$ and $H_B\succeq0$ each have nonempty nullspaces.
  Let $S = \Null(H_A)$ and suppose $H_B|_S$ is not the zero matrix. Then
  \begin{equation}
    \gamma(H)\ge\min\{\gamma(H_A),\gamma(H_B)\}\cdot \min_{\ket v\in\Gamma:\braket vv = 1}\braketb{v}{(\I-\Pi_B)}/2.
  \end{equation}
\end{corollary}
\begin{proof}
  Follows from \Cref{cor:GN:cor1} and the fact $1-\sqrt{1-x}\ge \frac x2$ for $x\in[0,1]$.
\end{proof}
\begin{corollary} \label{lem:redundant}
  Let $H_A \succeq 0, H_B\succeq0$ be Hamiltonians with $\Null(H_A)\subseteq\Null(H_B)$.
  Then $\gamma(H_A+H_B)\ge\min\{\gamma(H_A),\gamma(H_B)\}$.
\end{corollary}
\begin{proof}
  In \Cref{cor:GN:cor1}, we have $\Gamma = \Null(H_A)\cap \Null(H_A+H_B)^\bot = \Null(H_A)\cap \Null(H_A)^\bot = \{0\}$.
  Thus $d = \norm{\Pi_B\vert_\Gamma} = 0$ and $\gamma(H_A+H_B)\ge \min\{\gamma(H_A), \gamma(H_B)\}$.
\end{proof}

\section{2D Hamiltonian}\label{sec:2D-Hamiltonian}

The following theorem extracts the 2D Hamiltonian construction central to \cite{GN13} so that we can use it in conjunction with our own clock construction.
We give a complete proof with a simplified Hamiltonian construction and improved analysis that gives soundness $\Omega(\gamma(\Hclock^{(N)})/N^2)$, as opposed to $\Omega(\gamma(\Hclock^{(N)})/N^6)$ from \cite{GN13}.
Our proof is fully analytic, improving on the partially numeric analysis of ~\cite{GN13}.

Note, our construction is structurally almost the same as \cite{GN13}, which would also work in conjunction with our clock (see \Cref{sec:clock}), requiring only a slight modification to the $\Hinit$ and $\Hend$ terms.
Also, the application of \Cref{lem:connect} to their Hamiltonian is not as straightforward because there is no separate gadget for $1$-local gates (see \Cref{fig:HU}).
So, our contribution to the following theorem is a simplified proof and improved soundness bounds.

\begin{theorem}\label{thm:generic-clock}
  Suppose we are given Hamiltonian terms as follows:
  \begin{enumerate}[label=(\arabic*)]
    \item Clock Hamiltonian $\HclockN$ whose nullspace is spanned by \emph{clock states} $\ket{C_i}$ for $i=1,\dots,N$.
    \item Projectors $h_{i,i+1}(U)$ acting on $\HHcomp\otimes \HHclock$ such that
    \begin{equation}\label{eq:Piclock hii+1 Piclock}
    \begin{aligned}
      \left(\I\otimes \PiclockN\right) h_{i,i+1}(U) \left(\I\otimes \PiclockN\right) = c_1 \bigl(&(\I\otimes \ketbrab{C_i} + \I\otimes \ketbrab{C_{i+1}}) \\- &(U^\dagger\otimes\ketbra{C_i}{C_{i+1}} + U\otimes \ketbra{C_{i+1}}{C_i})\bigr)
    \end{aligned}
    \end{equation}
    for constant\footnote{It suffices for our purposes to have the same $c$ for all $i$. In principle, however, one can also allow different constants depending on the choice of index $i$.} $c_1$.
    We write $h_{i,i+1} := h_{i,i+1}(\I)$.
    \item Projectors $C_{\ge i}, C_{\le i}$ acting on $\HHclock$ such that
    \begin{align}
      \PiclockN C_{\ge i} \PiclockN &= \sum_{j=i}^N c_{2,i,j}\ketbra{C_j}{C_j},\label{eq:Piclock Cge Piclock}\\
      \PiclockN C_{\le i} \PiclockN &= \sum_{j=1}^i c_{3,i,j}\ketbra{C_j}{C_j},\label{eq:Piclock Cle Piclock}
    \end{align}
    and  $c_{2,i,j},c_{3,i,j}\ge c_2$ for all $i,j\in[N]$ for some constant $c_2$.
    \item All above projectors ($h_{i,i+1}(U), C_{\ge i}, C_{\le i}$) pairwise commute, besides $h_{i,i+1}(U)$ and $h_{i,i+1}(U')$ for non-commuting $U,U'$.
    For all $i\ne j$, $h_{i,i+1}(U)h_{j,j+1}(U')=0$.
    \item If $\Pi_1,\dots,\Pi_k$ are projectors of the form $C_{\le i},C_{\ge i}$, then $\bra{C_{j_1}}\Pi_{1}\dotsm\Pi_{k}\ket{C_{j_2}}=0$ for $j_1 \ne j_2$.%
    \footnote{(4) and (5) are only used to prove the soundness lower bound and may be dropped, decreasing soundness by a polynomial factor.}
  \end{enumerate}
  Then any problem in $\QMAo$ can be reduced to \QSAT with Hamiltonians $H$ acting on $\HHcomp \otimes \HHclocki \otimes \HHclockii$ (labeled $Z,X,Y$) with terms
  $\splitatcommas{(\HclockN)_X,(\HclockN)_Y,\Pi_Z\otimes (h_{i,i+1})_A, (C_{\sim j})_A\otimes (h_{i,i+1})_B,(h_{i,i+1}(U))_{Z,A}, (C_{\le i})_A\otimes (C_{\ge j})_B}$, where $A,B\in \{X,Y\}, A\ne B$, ${\sim}\in \{{\le},{\ge}\}$\footnote{Here we use `$\sim$' as a placeholder for a relation `$\le$' or `$\ge$'.}, $i,j\in[N]$, $\Pi_Z \in \{\ketbra00,\ketbra11\}$ is a single-qubit projector acting on $\HHcomp$, and $U$ is either a $1$-local gate from the $\QMAo$ circuit or $U\in \{\sigma^Z, B\}$.
  The soundness is $\Omega(\gamma(\Hclock^{(N)})/N^2)$, where $N=\Theta(g)$ and $g$ is the number of gates used by the $\QMAo$ verifier.
\end{theorem}
\noindent Note, this theorem is not explicitly stated in \cite{GN13}, but is implicitly proven, albeit with a soundness of $\Omega(\gamma(\HclockN/N^6))$.
Thus, as an immediate first consequence we can (using the clock Hamiltonian of \cite{GN13}) recover $\QMAo$-hardness of $3$-QSAT, but with improved soundness:
\thmThree*

\subsection{Two-qubit gates}
The main challenge in building a circuit-Hamiltonian using only the gadgets listed in \Cref{thm:generic-clock} is constructing $2$-qubit gates.
\begin{figure}[t]
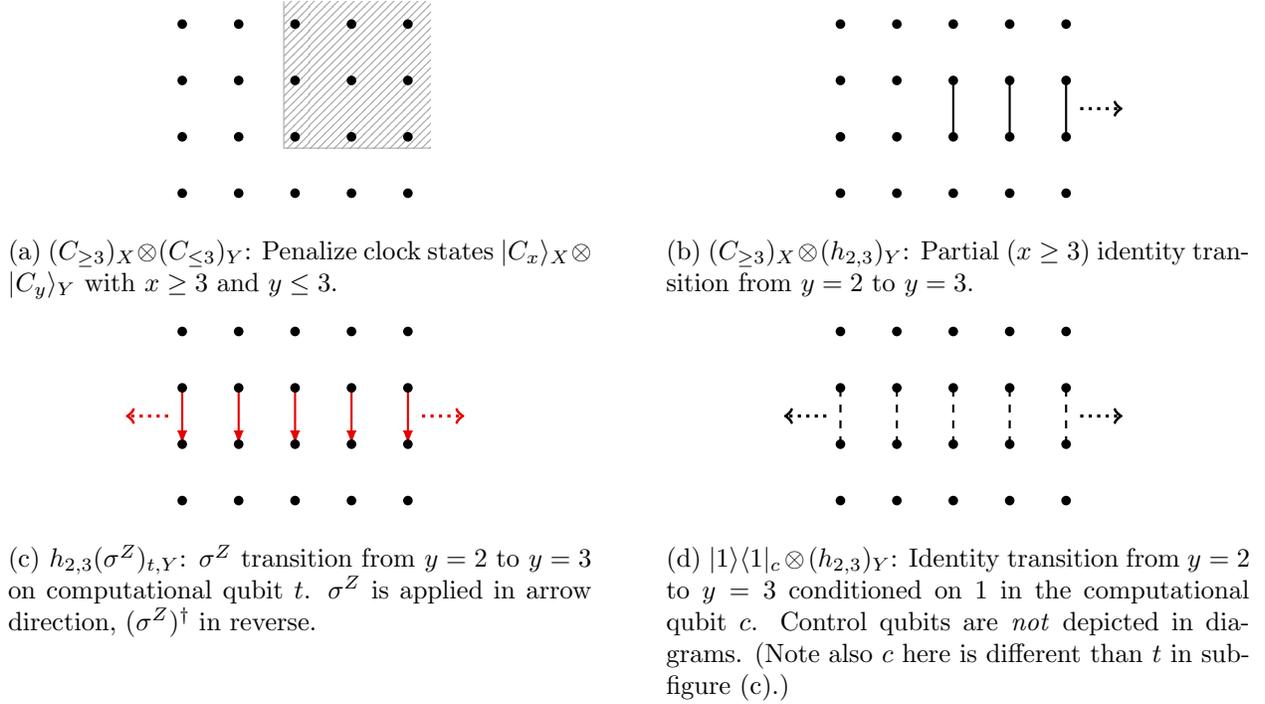

  \begin{subfigure}[t]{0.47\textwidth}
  \ifx\figcommon\undefined\input{common.tex}\fi

\begin{tikzpicture}
\node (v1) at (2.8,1.8) {};
\node (v2) at (5.4,4.4) {};
\path[exclude]  (v1) rectangle (v2);

\baseGrid

\draw[color=colExclude] (2.8,4.4) -- (v1.center) -- (5.4,1.8);
\end{tikzpicture}
  \centering
  \caption{$(C_{\ge3})_X\otimes (C_{\le 3})_Y$: Penalize clock states $\ket{C_x}_X\otimes\ket{C_y}_Y$ with $x\ge3$ and $y\le3$.}
  \label{fig:gadgets:quadrant}
  \end{subfigure}
  \hfill
  \begin{subfigure}[t]{0.47\textwidth}
  \centering
  \ifx\figcommon\undefined\input{common.tex}\fi

\begin{tikzpicture}

\foreach \i in {3,4,5} {
	\path[transition] (\i,2) edge (\i,3);
}

\baseGrid

\draw[arrow cont] (5.25,2.5) -- (6,2.5);
\end{tikzpicture}
  \caption{$(C_{\ge 3})_X\otimes (h_{2,3})_Y$: Partial ($x\ge3$) identity transition from $y=2$ to $y=3$.}
  \end{subfigure}
  \begin{subfigure}[t]{0.47\textwidth}
  \ifx\figcommon\undefined\input{common.tex}\fi

\begin{tikzpicture}

\foreach \i in {1,2,...,\baseGridW} {
	\draw[unitary,outer sep] (\i,3) -- (\i,2);
}

\baseGrid

\draw[arrow cont unitary] (0.75,2.5) -- (0,2.5);
\draw[arrow cont unitary] (5.25,2.5) -- (6,2.5);
\end{tikzpicture}
  \centering
  \caption{$h_{2,3}(\sigma^Z)_{t,Y}$: $\sigma^Z$ transition from $y=2$ to $y=3$ on computational qubit $t$. $\sigma^Z$ is applied in arrow direction, $(\sigma^Z)^\dagger$ in reverse.}
  \end{subfigure}
  \hfill
  \begin{subfigure}[t]{0.47\textwidth}
  \ifx\figcommon\undefined\input{common.tex}\fi

\begin{tikzpicture}
\foreach \i in {1,2,...,\baseGridW} {
	\path[conditional] (\i,2) edge (\i,3);
}

\baseGrid

\draw[arrow cont] (0.75,2.5) -- (0,2.5);
\draw[arrow cont] (5.25,2.5) -- (6,2.5);
\end{tikzpicture}
  \centering
  \caption{$\ketbra11_{c}\otimes (h_{2,3})_Y$: Identity transition from $y=2$ to $y=3$ conditioned on $1$ in the computational qubit $c$. Control qubits are \emph{not} depicted in diagrams. (Note also $c$ here is different than $t$ in subfigure (c).)}
  \end{subfigure}

  \caption{Graphical representations of constraints from the circuit Hamiltonian of \Cref{thm:generic-clock}.
  The $X$-axis goes from left to right, the $Y$-axis from top to bottom; both start counting at index $1$.
  Dots: Each dot represents a clock state (here a $5\times 4$ subspace of clock states is depicted).
  Edges: Transitions between two clock states.
  Red edges: Unitary transitions.
  Dashed edges: Conditional transitions (with control $\ketbra00$ horizontally, $\ketbra11$ vertically), i.e. transitions which only take place if the computation/control register reads $0$ (respectively, $1$).
  Dashed arrows: Indicate that edges continue in that direction on a larger clock space when combining gadgets.
  Hatched/shaded area: Penalized clock states, which are thus not in the Hamiltonian's nullspace.}
  \label{fig:gadgets}
\end{figure}
Gosset and Nagaj \cite{GN13} solve this by using the 2D clock to split the computation conditioned on a computational qubit by routing the states with $\ket0_c$ and $\ket1_c$ in the control register $c$ along different paths: Path 1 first applies $B=\frac1{\sqrt2}\begin{psmallmatrix}1&i\\i&1\end{psmallmatrix}$ and then $\sigma^Z=\begin{psmallmatrix}1&0\\0&-1\end{psmallmatrix}$ to the $\ket0_c$ states, and Path 2 first $\sigma^Z$ and then $B$ to $\ket1_c$.
This implements the gate $V = \ketbra 00 \otimes \sigma^ZB + \ketbra 11\otimes B\sigma^Z$, with $(T^2\otimes T^6 \hat H T^2)V = \CNOT$, where recall $\hat H$ is Hadamard (to avoid confusion with $H$ used elsewhere for ``Hamiltonian'').

\Cref{fig:gadgets} depicts the fundamental gadgets used to implement $V$.
$V$ is implemented by the Hamiltonian $H_V$ depicted in \Cref{fig:HV}.
When considering $H_V$ on its own in the $7\times 7$ grid, the arrows indicating continuations of the edges are not relevant.
However, to implement a full circuit we need to chain many copies of $H_V$ together, connecting the upper left of one copy to the lower right of the previous, which we formally justify via the ``Nullspace Connection Lemma'' \ref{lem:connect}.
In doing so, we must take care that the continued edges do not connect to another $H_V$ instance.
\begin{figure}[t]
  \centering
  \ifx\figcommon\undefined\input{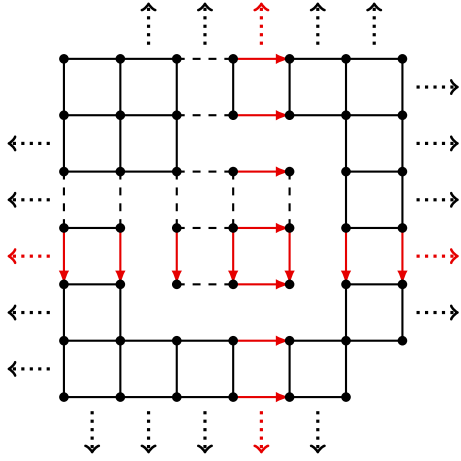}\fi

\begin{tikzpicture}

\makeEdgesDownXY{1}{1}{7}{transition}{arrow cont}

\makeEdgesDownXY{2}{6}{7}{transition}{arrow cont}
\makeEdgesUpXY{2}{1}{3}{transition}{arrow cont}

\makeEdgesDownXY{3}{6}{7}{transition}{arrow cont}
\makeEdgesUpXY{3}{1}{5}{conditional}{arrow cont}

\makeEdgesUpDownXY{4}{1}{7}{unitary}{arrow cont unitary}

\makeEdgesDownXY{5}{6}{7}{transition}{arrow cont}
\makeEdgesUpXY{5}{1}{2}{transition}{arrow cont}

\makeEdgesUpXY{6}{1}{6}{transition}{arrow cont}

\makeGridBR{7}{7}

\end{tikzpicture}
  \caption{Graphical representation of the Hamiltonian $H_V$ implementing the $V$ gate with control $c$ and target $t$. Dashed edges: $\ketbra00_{c}\otimes (h_{3,4})_X + \ketbra11_{c}\otimes (h_{3,4})_Y$. Note that technically there are also dashed edges $(3,4)$ in rows $6$ and $7$, but these are not drawn, since the identity transitions on the same edges (drawn as solid edges) also act here. Red edges: $h_{4,5}(B)_{t,X} + h_{4,5}(\sigma^Z)_{t,Y}$, i.e. $B$ applied on red edges along the $X$-axis/horizontally, $\sigma^Z$ on red edges on the $Y$-axis/vertically.}
  \label{fig:HV}
\end{figure}
Next, we analyze the nullspace of $H_V$ and prove that it has the structure required by \Cref{lem:connect}.
This analysis is in the clock space and disregards the constants of \Cref{eq:Piclock hii+1 Piclock,eq:Piclock Cge Piclock,eq:Piclock Cle Piclock} as they do not affect the nullspace.
The definition of $H_V$ is evident from \Cref{fig:HV}, but for completeness we state it formally (with control $c$ and target $t$ from the compute register $Z$):
\begin{equation}
  \begin{aligned}
    H_V &= h_{1,2}\otimes C_{\ge 1} + h_{2,3}\otimes(C_{\le3} + C_{\ge 6}) + \ketbra00_{c}\otimes (h_{3,4})_{X} + h_{3,4}\otimes C_{\ge 6} \\
    &\qquad+ h_{4,5}(B)_{t,X} + h_{5,6}\otimes(C_{\le2} + C_{\ge 6}) + h_{6,7}\otimes C_{\le 6}\\
    &+C_{\ge 1} \otimes h_{1,2} + (C_{\le3} + C_{\ge 6})\otimes h_{2,3} + \ketbra11_{c}\otimes(h_{3,4})_{Y} +  C_{\ge 6}\otimes h_{3,4} \\
    &\qquad+ h_{4,5}(\sigma^Z)_{t,Y} + (C_{\le2} + C_{\ge 6})\otimes h_{5,6} +  C_{\le 6}\otimes h_{6,7}
  \end{aligned}
\end{equation}
Note that $H_V$ is almost symmetric with respect to the $X$- and $Y$-axes, the only difference being $B$ in the $X$-axis and $\sigma^Z$ in the $Y$-axis.
To derive the nullspace of $\Null(H_V)$, we deviate from \cite{GN13} and give an analytical characterization (versus numeric) via unitary labelled graphs (ULG) (\Cref{def:ULG})\cite{BCO17}. To set this up, we decompose $H_V$ so that we can handle the conditional edges in a later step:
$H_V = \ketbra00_c \otimes H_{V,0} + \ketbra11_c \otimes H_{V,1}$ with $H_{V,0},H_{V,1}$ as depicted in \Cref{fig:HV-decomp}.
The graphical representation makes the two different paths for the $\ket{0}_c$ and $\ket{1}_c$ states evident, and one can see that they traverse the $B$ ($X$-axis) and $\sigma_Z$ ($Y$-axis) edges in a different order.
\begin{figure}[t]
  \begin{subfigure}[t]{0.47\textwidth}
    \centering
    \ifx\figcommon\undefined\input{common.tex}\fi

\begin{tikzpicture}

\draw[component] (0.8,7.2) -- (4.2,7.2) -- (4.2,4.8) -- (0.8,4.8) -- cycle;
\draw[component] (2.2,3.2) -- (2.2,2.2) -- (4.2,2.2) -- (4.2,0.8) -- (0.8,0.8) -- (0.8,3.2) -- cycle;
\draw[component] (4.8,2.2) -- (4.8,0.8) -- (6.2,0.8) -- (6.2,1.8) -- (7.2,1.8) -- (7.2,3.2) -- (5.8,3.2) -- (5.8,2.2) -- cycle;
\draw[component] (4.8,7.2) -- (7.2,7.2) -- (7.2,3.8) -- (5.8,3.8) -- (5.8,4.8) -- (4.8,4.8) -- cycle;
\draw[component] (0.8,4.2) -- (0.8,3.8) -- (2.2,3.8) -- (2.2,4.2) -- cycle;

\makeEdgesXY{1}{1}{7}{transition}

\makeEdgesXY{2}{6}{7}{transition}
\makeEdgesXY{2}{1}{3}{transition}

\makeEdgesXY{3}{6}{7}{transition}

\makeEdgesX{3}{1}{5}{transition}

\makeEdgesXY{4}{1}{7}{unitary}

\makeEdgesXY{5}{6}{7}{transition}
\makeEdgesXY{5}{1}{2}{transition}

\makeEdgesXY{6}{1}{6}{transition}

\makeGridBR{7}{7}

\node[anchor=north east] at (0.8,7.2) {$\TL_0$};
\node[anchor=north west] at (7.2,7.2) {$\TR_0$};
\node[anchor=south east] at (0.8,0.8) {$\BL_0$};
\node[anchor=south west] at (7.2,1.8) {$\BR_0$};

\node[anchor=east] at (0.8,4) {$\M_0$};
\end{tikzpicture}
    \caption{$H_{V,0}$}
  \end{subfigure}
  \hfill
  \begin{subfigure}[t]{0.47\textwidth}
    \centering
    \ifx\figcommon\undefined\input{common.tex}\fi

\begin{tikzpicture}

\draw[component] (0.8,7.2) -- (3.2,7.2) -- (3.2,3.8) -- (0.8,3.8) -- cycle;
\draw[component] (3.2,3.2) -- (3.2,2.2) -- (4.2,2.2) -- (4.2,0.8) -- (0.8,0.8) -- (0.8,3.2) -- cycle;

\draw[component] (4.8,2.2) -- (4.8,0.8) -- (6.2,0.8) -- (6.2,1.8) -- (7.2,1.8) -- (7.2,3.2) -- (5.8,3.2) -- (5.8,2.2) -- cycle;
\draw[component] (4.8,7.2) -- (7.2,7.2) -- (7.2,3.8) -- (5.8,3.8) -- (5.8,5.8) -- (4.8,5.8) -- cycle;
\draw[component] (3.8,7.2) -- (3.8,5.8) -- (4.2,5.8) -- (4.2,7.2) -- cycle;

\makeEdgesXY{1}{1}{7}{transition}

\makeEdgesXY{2}{6}{7}{transition}
\makeEdgesXY{2}{1}{3}{transition}

\makeEdgesXY{3}{6}{7}{transition}

\makeEdgesY{3}{1}{5}{transition}

\makeEdgesXY{4}{1}{7}{unitary}

\makeEdgesXY{5}{6}{7}{transition}
\makeEdgesXY{5}{1}{2}{transition}

\makeEdgesXY{6}{1}{6}{transition}

\makeGridBR{7}{7}

\end{tikzpicture}
    \caption{$H_{V,1}$}
  \end{subfigure}
  \caption{Decomposition of $H_V = \ketbra00_c \otimes H_{V,0} + \ketbra11_c \otimes H_{V,1}$.}
  \label{fig:HV-decomp}
\end{figure}
\begin{lemma}\label{lem:conditionHV}
  $\Pi_K H_V\Pi_K$ satisfies (2) from \Cref{lem:connect} with $U_i=V_{c,t}$, $\ket{u_i} = \ket{1,1}_{XY}$, and $\ket{v_i} = \ket{7,6}_{XY}$, where $\Pi_K$ is the projector onto the clock space\footnote{For clarity, $K$ is the label set for the vertices in \Cref{fig:HV}. $(7,7)$ is omitted from $K$, as depicted by the missing vertex in the bottom right corner of \Cref{fig:HV}.} $K:=[7]^2\setminus\{(7,7)\}$.
\end{lemma}
\begin{proof}
  In the following, we apply the projector onto the clock space $\Pi_K$ implicitly.
  Then $H_{V,0}$ and $H_{V,1}$ are unitary labelled graphs (ULG, \Cref{def:ULG}).
  Their respective center components are frustrated because they contain (by flipping two edges and inverting their unitaries using \cite[Proposition 39]{BCO17}) a directed cycle $B\sigma^ZB^\dagger\sigma^Z\ne\I$.
  Hence, their nullspaces have only support on the outer connected component, which is \emph{simple} (\Cref{def:simple}, i.e. the product of unitaries on any cycle is $\I$).
  To see that a directed cycle $\ne\I$ is necessarily frustrated, remove an edge so that the resulting ULG is simple.\footnote{Properties of frustrated ULGs are also discussed in \cite{BC18}.}
  Then the nullspace of that ULG \cite[Lemma 42]{BCO17} is not in the nullspace of the removed edge.

  Let $Z'$ be the computational space $Z$ without the control qubit $c$, $\{\ket{\alpha'_j}\}_{j\in [d_Z/2]}$ an orthonormal basis of $Z'$ ($d_Z$ the dimension of $Z$). Define for $H_{V,0}$ the states $\ket{\mathcal{\TL}_0},\ket{\TR_0},\ket{\BL_0},\ket{\BR_0},\ket{\M_0}$ (denoting, e.g., $\TL$ for ``top-left'') and for $H_{V,1}$ states $\ket{\mathcal{\TL}_1},\ket{\TR_1},\ket{\BL_1},\ket{\BR_1},\ket{\M_1}$  as the non-normalized uniform superpositions\footnote{For example, $\ket{\mathcal{TL}_0}$ is an equal superpositions over all vertices/clock states in \Cref{fig:HV-decomp}(a)'s top-left grey-shaded block.} over the corresponding clock states in \Cref{fig:HV-decomp}.
  By \cite[Lemma 42]{BCO17}, $\Null(H_{V,0})$ and $\Null(H_{V,1})$ are respectively spanned by vectors given by, for $j\in[d_Z/2]$:
  \begin{subequations}
  \begin{align}
    \ket{\psi_0(\alpha'_j)} &= \ket{\alpha'_j}\ket{\TL_0} + B_t\ket{\alpha'_j}\ket{\TR_0}+ (\sigma^ZB)_t\ket{\alpha'_j}\ket{\BR_0} \\
    &\quad+(B^\dagger\sigma^ZB)_{t}\ket{\alpha'_j}\ket{\BL_0} +(\sigma^ZB^\dagger\sigma^ZB)_{t}\ket{\alpha'_j}\ket{\M_0},\nonumber\\
    \ket{\psi_1(\alpha'_j)} &= \ket{\alpha'_j}\ket{\TL_1} + \sigma^Z_t\ket{\alpha'_j}\ket{\BL_1}+ (B\sigma^Z)_t\ket{\alpha'_j}\ket{\BR_1} \\
    &\quad+(\sigma^ZB\sigma^Z)_{t}\ket{\alpha'_j}\ket{\TR_1} +(B^\dagger\sigma^ZB\sigma^Z)_{t}\ket{\alpha'_j}\ket{\M_1},\nonumber
  \end{align}
  \end{subequations}
  Thus $\Null(H_V)$ is spanned by
  \begin{equation}
    \ket{\psi(\alpha_{b,j})} = \ket{b}_c\otimes \ket{\psi_b(\alpha'_j)}_{Z'XY}
  \end{equation}
  for $b\in\bin$, $j\in[d_Z/2]$, where $\{\ket{\alpha_{b,j}}\}_{b\in\bin,j\in [d_Z/2]}$ with $\ket{\alpha_{b,j}} = \ket{b}_{c}\ket{\alpha'_{j}}_{Z'}$ is an orthonormal basis of $Z$.
  Hence, $\Null(H_V)$ satisfies the property (2a) of \Cref{lem:connect}.

  Next, the function $L$ with $L\ket{\alpha_{b,j}} = \ket{\psi(\alpha_{b,j})}$ is clearly linear and $L^\dagger L = \lambda\I$ for $\lambda\in O(1)$, which implies (2b).
  Property (2c) holds by definition, and (2d) is trivial.

  Finally, $\bra{1,1}_{XY}\ket{\psi(\alpha_{b,j})} = \ket{b}_c\ket{\alpha'_j}_{Z'} = \ket{\alpha_{b,j}}_Z$ and
  \begin{equation}
    \bra{7,6}_{XY}\ket{\psi(\alpha_{b,j})} = \left\{\!\begin{aligned}
      &\ket{0}_c\otimes (\sigma^ZB)_t\ket{\alpha'_j}_{Z'},& b=0\\[1ex]
      &\ket{1}_c\otimes (B\sigma^Z)_t\ket{\alpha'_j}_{Z'},& b=1
    \end{aligned}\right\}
    = V_{c,t}\ket{\alpha_{b,j}},
  \end{equation}
  implying (2e) and (2f). This completes the proof.
\end{proof}

\subsection{Single-qubit gates}
Thus far, we have discussed the implementation of {two}-qubit unitaries. Next, {single}-qubit unitaries $U$ are implemented with a simpler gadget $H_{U}$, depicted in \Cref{fig:HU}. $H_U$ also has the necessary properties for \Cref{lem:connect}; this is stated in \Cref{lem:conditionHU}. Its proof is analogous to \Cref{lem:conditionHV}, and thus omitted.
\begin{figure}
  \centering
  \ifx\figcommon\undefined\input{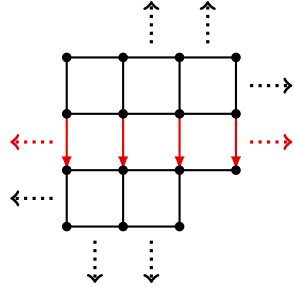}\fi
{
\renewcommand\dC{5}
\begin{tikzpicture}

\makeEdgesDownXY{1}{1}{4}{transition}{arrow cont}
\makeEdgesUpXY{3}{1}{3}{transition}{arrow cont}

\makeEdgesX{2}{1}{4}{transition}
\makeArrowUpX{2}{1}{arrow cont}
\makeArrowDownX{2}{4}{arrow cont}

\makeEdgesY{2}{1}{4}{unitary}
\makeArrowUpY{2}{1}{arrow cont unitary}
\makeArrowDownY{2}{4}{arrow cont unitary}

\makeGridBR{4}{4}

\end{tikzpicture}
}
  \caption{Graphical representation of the Hamiltonian $H_U$ implementing a single-qubit unitary $U$ on computational qubit $z$.}
  \label{fig:HU}
\end{figure}
\begin{equation}
\begin{aligned}
  H_U &= (h_{1,2})\otimes C_{\ge 1} + (h_{2,3})_X + (h_{3,4})\otimes C_{\le 3}\\
    &\,+ C_{\ge 1}\otimes (h_{1,2}) + h_{2,3}(U)_{z,Y} + C_{\le 3}\otimes(h_{3,4})
\end{aligned}
\end{equation}
\begin{lemma}\label{lem:conditionHU}
  $\Pi_K H_{U}\Pi_K$ satisfies (2) from \Cref{lem:connect} with $U_i=U_z$, $\ket{u_i} = \ket{1,1}_{XY}$, and $\ket{v_i} = \ket{4,3}_{XY}$, where $\Pi_K$ is the projector onto the clock space $K:=[4]^2\setminus\{(4,4)\}$.
\end{lemma}

\subsection{Combining gadgets to build the full Hamiltonian}
\begin{figure}[t]
  \centering
  \includegraphics[width=0.7\textwidth]{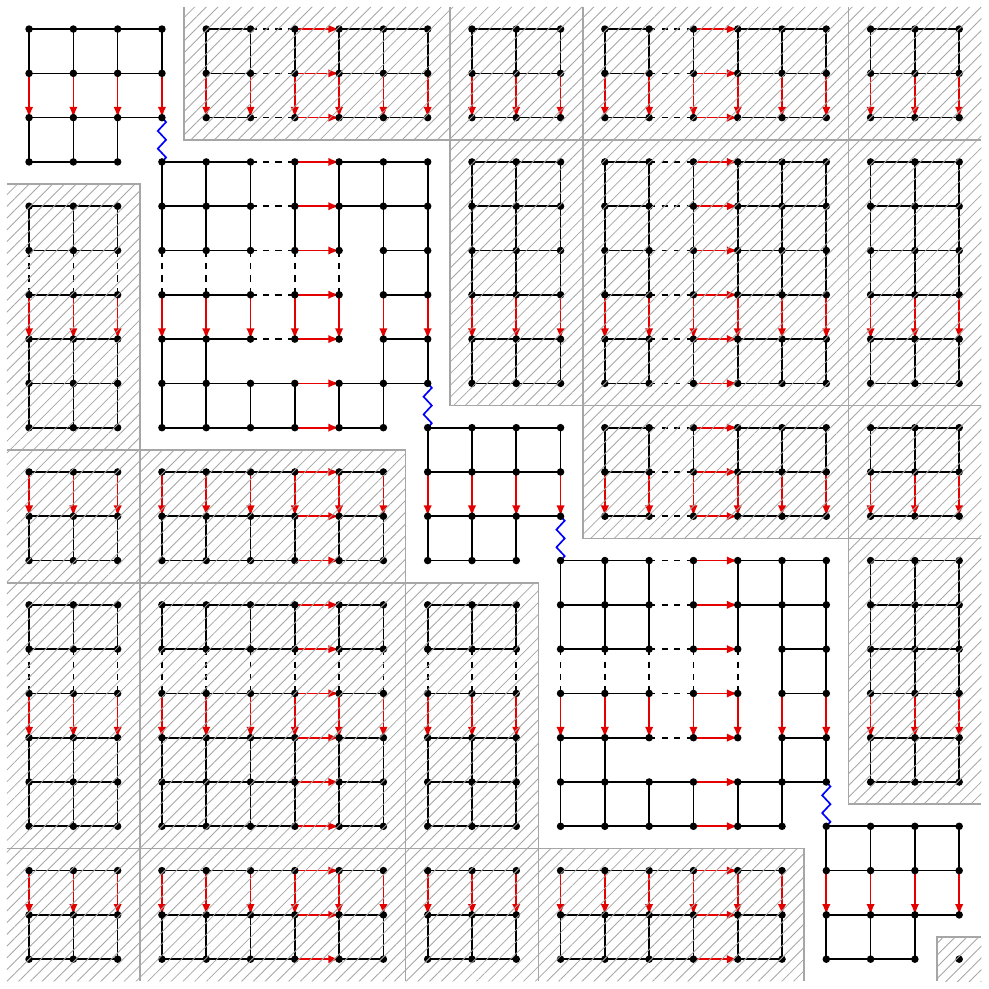}
  \caption{Graphical representation of $\Hdiag + \Hgate+\Hlink$ from \Cref{eq:Hdiag,eq:Htrans1,eq:Htrans2} for $T=2$. Blue zigzag edges represent the new edges from $\Hlink$ not present in $\Hgate$.}
  \label{fig:H}
\end{figure}
To build the full circuit Hamiltonian, we place multiple $H_V$ and $H_U$ gadgets on a diagonal and constraints to penalize the clock states outside these gadgets, as depicted in \Cref{fig:H}.
Let the full circuit consist of gates\footnote{Note we are using $T$ for the number of gates here, whereas \Cref{thm:generic-clock} used $g$ for the number of gates of the $\QMAo$ verifier. This is because in converting the $\QMAo$ verifier into the gadget framework here, we will have $T>g$, i.e. a poly-size blowup in gate count.} $G_1,\dots,G_T$ such that each $G_i$ is either a $V$ gate (in which case we set a counter $s_i:=6$) or a one-qubit unitary (set counter $s_i:=3$).
$V$ acts on $n=n_a+n_p$ qubits, $n_a$ ancilla and $n_p$ proof qubits.
The classical input is assumed to already be embedded inside $V$.
Define the \emph{offset} of each gadget (from the top-left of \Cref{fig:H}) as $t_1:=1$ and
\begin{equation}
  t_i = \sum_{j=1}^{i-1}s_i.
\end{equation}
The first $T/3$ and the last $T/3$ gates shall be single-qubit identity gates, so that we can achieve $\Omega(1/T^2)$ soundness via the ``idling trick'' where the input and output state are repeated on a constant fraction of timesteps.
The full Hamiltonian is then defined as $H = \Hclock' + \Hdiag + \Hgate + \Hlink + \Hin + \Hout$ (see \Cref{fig:H}) with
\begin{subequations}
  {\allowdisplaybreaks
  \begin{align}
    \Hclock' &=(\Hclock)_X + (\Hclock)_Y\label{eq:Hclock'}\\
    \Hdiag &= \sum_{i=2}^T (C_{\ge t_i + 2}\otimes C_{\le t_i} + C_{\le t_i}\otimes C_{\ge t_i+2}) \label{eq:Hdiag} + C_{\ge t_{T+1}}\otimes C_{\ge t_{T+1}}\\
    \Hgate &= \sum_{i=1}^{T} H_{G_i}^{(t_i)} \label{eq:Htrans1}\\
    \Hlink &= \sum_{i=2}^{T} \left(h_{t_i,t_i+1}\otimes C_{\le t_i+1}\right)\label{eq:Htrans2}\\
    \Hin &= \sum_{i=1}^{n_a} \ketbra11_{Z_{i}}\otimes (C_{\le t_{T/3}})_Y\\
    \Hout &= \ketbra00_{Z_{1}}\otimes (C_{\ge t_{2T/3}+1})_Y\label{eq:Hout}
  \end{align}}
\end{subequations}
where the notation $H_{V}^{(j)}$ indicates that the $H_V$ gadget is shifted by $j$ in both $X$- and $Y$-direction.
From $\Hlink$, we only require the terms $h_{t_i,t_i+1}\otimes\ketbrab{C_{t_i+1}}$.
However, these terms cannot be directly implemented using just the terms allowed in \Cref{thm:generic-clock}.
Therefore, $\Hlink{}$ adds many edges already present in $\Hlink$, which can be disregarded in subsequent analysis due to \Cref{lem:redundant}.
Further, the edges outside the squares on the diagonal correspond to the arrows from \Cref{fig:HV,fig:HU} indicating that the edges continue on arrow direction when embedded in a larger clock space.
As can be seen in \Cref{fig:H}, there are no edges between the hatched areas and the squares corresponding to $H_V,H_U$ gadgets, so the continued edges cause no issues.

\begin{proof}[Proof of \Cref{thm:generic-clock}]
  $\gamma(\Hclock') = \gamma(\Hclock)$ and the other Hamiltonians \eqref{eq:Hdiag} to \eqref{eq:Hout} have $\gamma(\cdot) = \Theta(1)$ because they are sums of commuting projectors.

  Let $H_1 = \Hclock' + \Hdiag$.
  We apply \Cref{cor:GN:cor1} to show $\gamma(H_1)\ge \Omega(\gamma(\Hclock))$.
  Define
  \begin{equation}
    S_1:=\Null(H_1) = \HHcomp \otimes \bigoplus_{i=1}^T\KK_{i},
  \end{equation}
  which is the span of the clock states of the $H_U$ and $H_V$ gadgets (i.e. the dots in \Cref{fig:H} outside the hatched areas) with $\KK_{G_i} = \Span\{\ket{C_x}\ket{C_y} \mid (x,y) \in K_{i}\}$ and
  \begin{equation}
    K_{i} = \{t_i + 1,...,t_i+s_i+1\}^2\setminus\{(t_i+s_i+1,t_i+s_i+1)\}.
  \end{equation}
  Then we have $\Gamma = \Span\{\ket{C_x}\ket{C_y}\mid (x,y)\in \overline K\}$, where $\overline K = (\bigcup_{i=1}^TK_i)^C$ is the set of penalized clock states.
  Since the summands of $\Hdiag$ are commuting projectors, we can write $\Pidiag = \prod_{i=1}^k(\I-H_i)$, where $H_1,\dots,H_k$ denote the projectors of $\Hdiag$ as given in \Cref{eq:Hdiag}.
  Let $\ket{v}\in\Gamma$, which can be written as $\ket{v} = \sum_{(x,y)\in\overline K} c_{xy}\ket{C_{x},C_{y}}$.
  Then
  \begin{subequations}
  \begin{align}
    \braketb{v}{\Pidiag} &= \sum_{(x,y)\in\overline K}\abs{c_{xy}}^2\braketb{C_x,C_y}{\Pidiag}\label{eq:bound-Pidiag:a}\\
    &\le \sum_{(x,y)\in\overline K} \abs{c_{xy}}^2\braketb{C_x,C_y}{(\I - H_{i_{xy}})} \le 1-\Theta(1),\label{eq:bound-Pidiag:b}
  \end{align}
  \end{subequations}
  where \eqref{eq:bound-Pidiag:a} holds because  $\bra{C_{x},C_{y}}\Pidiag\ket{C_{x'},C_{y'}} = 0$ if $(x,y) \ne (x',y')$ by condition (5) of \Cref{thm:generic-clock} and \eqref{eq:bound-Pidiag:b} follows from the observation that for each $(x,y)\in\overline K$, there exists $i_{xy}\in [k]$ such that $\braketb{C_x,C_y}{H_{i_{xy}}}\ge \Theta(1)$.

  Let $H_2 = H_1 + \Hgate + \Hlink$.
  By \Cref{cor:GN:cor2}, we get
  \begin{equation}
    \gamma(H_2)\ge \Omega\left(\gamma(\Hclock)\cdot \frac{\gamma((\Hgate + \Hlink)\vert_{S_1})}{\norm{(\Hgate + \Hlink)\vert_{S_1}}}\right),
  \end{equation}
  where $\norm{\Hgate\vert_{S_1}} = \Theta(1)$ as $\Hgate\vert_{S_1}$ is block diagonal with constant-size blocks (squares on the diagonal of \Cref{fig:H} without the zigzag edges).
  $\norm{\Hlink\vert_{S_1}} = \Theta(1)$ as the $h_{i,i+1}$ terms of $\Hlink$ operate on separate rows.

  By \Cref{lem:conditionHV,lem:conditionHU}, $\Null(\Hgate\vert_{S_1}) = \Span\{\ket{\psi_{G_i}(\alpha)}\mid i\in [T], \ket\alpha\in\HHcomp\} =: S_2$, where the $\ket{\psi_{G_i}(\alpha)}$ are defined as $\ket{\psi_i(\alpha)}$ in (2) of \Cref{lem:connect} with Hamiltonians $H_{G_i}\vert_{S_1}$.

  Since $\Hlink\vert_{S_1}$ contains ``redundant terms'' that are already contained in $\Hgate\vert_{S_2}$, define $\Hlink'$ as $\Hlink\vert_{S_1}$ without these edges (i.e. only the zigzag edges in \Cref{fig:H}).
  Thus $\Hlink' = \sum_{i=2}^{T}h'_{t_i,t_i+1}\otimes \ketbrab{C_{t_i+1}}$, where $h'_{i,i+1}(U) = \I\otimes \ketbrab{C_i} + \I\otimes \ketbrab{C_{i+1}} - U^\dagger\otimes\ketbra{C_i}{C_{i+1}} - U\otimes \ketbra{C_{i+1}}{C_i}$.
  Then $\Htrans := \Hgate\vert_{S_1} + \Hlink\vert_{S_1}$ and $\Htrans':=\Hgate\vert_{S_1} + \Hlink'$ have the same nullspace and $\gamma(\cdot)$ (up to a constant factor).
  \Cref{lem:connect} then gives $\gamma(\Htrans') = \Omega(1/T^2)$ and $\Null(\Htrans') = \Span\{\ket{\psihist(\alpha)}\mid\ket\alpha\in\HHcomp\}$, with $\ket{\psihist(\alpha)} = \sum_{i=1}^{T}\ket*{\vphantom{\big(}\psi_{G_i}(G_{i-1}\dotsm G_1 \ket{\alpha})}$.
  Hence, $\gamma(H_2) = \Omega(\gamma(\Hclock)/T^2)$.

  Let $H_3 = H_2 + \Hin$.
  Then $\Null(H_3) = \Span\{\ket{\psihist(0^{n_a}x)}\mid x\in\bin^{n_p}\}$.
  Applying \Cref{cor:GN:cor1} we have $\Gamma = \Span\{\ket{\psihist(x_ax_p)}\mid x_a\in\bin^{n_a}\setminus\{0^{n_a}\},x_p\in\bin^{n_p}\}$.
  Then for any $\ket v\in\Gamma$, $\braketb{v}{\Piin} \le 1 - \Theta(1)$, where $\Piin = (\sum_{x\in\bin^{n_a}}\ketbra{x}{x}_{Z_{1\dots n_a}})\otimes (\I - C_{\le t_{T/3}})_Y$
  is the projector onto the nullspace of $\Hin$, because $\Piin\ket{\psi_{G_i}(\psi_i)}=0$ for $i\in[T/3]$, $\ket{\psi_i} = G_{i-1}\dotsm G_1\ket{x_a}\ket{\psi}$, and $x_a\ne0^{n_a}$.
  Hence $\gamma(H_3) = \Omega(\gamma(\Hclock)/T^2)$.

  Finally, we have $H = H_3 + \Hout$.
  In the YES-case, the circuit accepts some proof $\ket{\psi}\in\CC^{2^{n_p}}$ with probability $1$ and thus $\ket{\psihist(\psiin)}\in\Null(H)$ with $\ket{\psiin}=\ket{0^{n_a}}\ket{\psi}$.
  In the NO-case, the circuit rejects with probability $\ge 2/3$.
  Hence, $\Null(H) = \{0\}$ and we can apply the Geometric Lemma \ref{lem:geometric}.
  The output state is repeated in the last $T/3$ gates, and so $\braketb{\psihist(\psiin)}{\Piout} \le 1-\Theta(1)$ for any proof $\ket\psi$.
  Thus, $\lmin(H) = \Omega(\gamma(\Hclock)/T^2)$.
\end{proof}

\section{\texorpdfstring{(2,5)-QSAT is QMA$_{\text{1}}$}{(2,5)-QSAT is QMA1}-complete}\label{sec:clock}

In this section, we define the clock to prove the $\QMAo$-completeness of $(2,5)$-QSAT (\Cref{thm:25}) and $(4,3)$-QSAT (Theorem \ref{thm:34}) using \Cref{thm:generic-clock}.
The main difficulty was the construction of a suitable clock for \Cref{thm:generic-clock} using just $(2,5)$-projectors.
The hardness of $(3,4)$-QSAT is obtained almost for free as part of the proof.
Since our construction may seem somewhat arbitrary, we give an intuition in \Cref{sec:intuition} by sketching a proof for the $\QMAo$-hardness of $(2,7)$-QSAT.

\subsection{Clock Hamiltonian definition}

We begin by defining several logical states.
A $5$-dit and a qubit are combined as follows to construct a logical $4$-dit (normalization factors are omitted):
\newcommand{\LU}{\mathbf{U}}
\newcommand{\LA}{\mathbf{A}}
\newcommand{\LD}{\mathbf{D}}
\newcommand{\Hlogical}{{H_{\mathrm{logical}}}}
\begin{subequations}\label{eq:logical4}
  \begin{align}
    \ket{\LU} &= \ket{u,0} + \ket{u',1}\\
    \ket{\LA_1} &= \ket{a_1,0}\\
    \ket{\LA_2} &= \ket{a_2,0}\\
    \ket{\LD} &= \ket{d,0}
  \end{align}
\end{subequations}
Here the standard basis states of the $5$-qudit are labeled $u,a_1,a_2,d,u'$.
The labels $u,a,d$ maybe be interpreted as ``unborn'', ``alive'', and ``dead'', respectively, following the convention of \cite{ER08}.
\begin{restatable}{lemma}{lemmaHlogicalIV}\label{lem:HlogicalIV}
  We have $\splitatcommas{\Null(\Hlogical_4)=\Span\{\ket{\LU},\ket{\LA_1},\ket{\LA_2},\ket{\LD}\}}$ for
  \begin{equation}
    \begin{aligned}
      \Hlogical_4 =\; &\ketbrab{u,1} + \ketbrab{u',0} + \ketbrab{a_1,1} + \ketbrab{a_2,1} + \ketbrab{d,1} + \phantom{a} \\&\ketbraa{(\ket{u,0}-\ket{u',1})}/2
    \end{aligned}
  \end{equation}
\end{restatable}
A $5$-dit and two qubits are combined to construct a logical qutrit:
\newcommand{\Lu}{\mathbf{u}}
\newcommand{\La}{\mathbf{a}}
\newcommand{\Ld}{\mathbf{d}}
\begin{subequations}\label{eq:logical3}
  \begin{align}
    \ket{\Lu} &= \ket{u,1,0} + \ket{x,0,0} \\
    \ket{\La} &= \ket{a,0,0} + \ket{x,1,0} \\
    \ket{\Ld} &= \ket{d,0,0} + \ket{d',0,1}
  \end{align}
\end{subequations}
Here the standard basis states of the $5$-qudit are labeled $u,a,d,x,d'$.
\begin{restatable}{lemma}{lemmaHlogicalIII}\label{lem:HlogicalIII}
We have $\splitatcommas{\Null(\Hlogical_3)=\Span(\ket{\Lu},\ket{\La},\ket{\Ld})}$ for
\begin{equation}
  \begin{aligned}
    \Hlogical_3 =\;
    &\ketbrab{u,0}_{\alpha\beta} + \ketbrab{u,1}_{\alpha\gamma} + \phantom{a}\\
    &\ketbrab{a,1}_{\alpha\beta} + \ketbrab{a,1}_{\alpha\gamma} + \phantom{a}\\
    &\ketbrab{d,1}_{\alpha\beta} + \ketbrab{d,1}_{\alpha\gamma} + \phantom{a}\\
    &\ketbrab{d',1}_{\alpha\beta} + \ketbrab{d',0}_{\alpha\gamma} + \phantom{a}\\
    &\ketbrab{x,1}_{\alpha\gamma} + \phantom{a}\\
    &\ketbraa{(\ket{x,0} - \ket{u,1})}_{\alpha\beta}/2 + \phantom{a}\\
    &\ketbraa{(\ket{x,1} - \ket{a,0})}_{\alpha\beta}/2 + \phantom{a}\\
    &\ketbraa{(\ket{d,0} - \ket{d',1})}_{\alpha\gamma}/2,
  \end{aligned}
\end{equation}
where the three qudits are labeled $\alpha,\beta,\gamma$.
\end{restatable}
\begin{proof}
  In \Cref{sec:proofs}.
\end{proof}
The main ``feature'' of these logical qudits is the fact that a $1$-local qubit projector suffices to identify the states $\ket{\LU}$, $\ket{\Ld}$, ``$\ket{\Lu}$ or $\ket{\La}$'', and implement a transition between $\ket{\Lu},\ket{\La}$.

We now define the clock states $\ket{C_i} := \ket{C_{i,1}} + \ket{C_{i,2}} + \ket{C_{i,3}} + \ket{C_{i,4}}$, where the states $\ket{C_{i,j}}$ are defined as follows:
\newcommand\ketarray[1]{\ket{\begin{rmatrix}[8mm]#1\end{rmatrix}}}
\begin{subequations}
  \begin{align}
    \ket{C_{1,1}} &= \ketarray{\Lu\LU & \Lu\LU & \dotsm & \Lu\LU & \Lu\LU}\\
    \ket{C_{1,2}} &= \ketarray{\La\LU & \Lu\LU & \dotsm & \Lu\LU & \Lu\LU}\\
    \ket{C_{1,3}} &= \ketarray{\Ld\LU & \Lu\LU & \dotsm & \Lu\LU & \Lu\LU}\\
    \ket{C_{1,4}} &= \ketarray{\Ld\LA_1 & \Lu\LU & \dotsm & \Lu\LU & \Lu\LU}\\[3mm]
    \ket{C_{2,1}} &= \ketarray{\Ld\LA_2 & \Lu\LU & \dotsm & \Lu\LU & \Lu\LU}\\
    \ket{C_{2,2}} &= \ketarray{\Ld\LD & \La\LU & \dotsm & \Lu\LU & \Lu\LU}\\
    \ket{C_{2,3}} &= \ketarray{\Ld\LD & \Ld\LU & \dotsm & \Lu\LU & \Lu\LU}\\
    \ket{C_{2,4}} &= \ketarray{\Ld\LD & \Ld\LA_1 & \dotsm & \Lu\LU & \Lu\LU}\\
    &\;\,\vdots\nonumber\\
    \ket{C_{N-1,1}} &= \ketarray{\Ld\LD & \Ld\LD & \dotsm & \Lu\LU & \Lu\LU}\\
    \ket{C_{N-1,2}} &= \ketarray{\Ld\LD & \Ld\LD & \dotsm & \La\LU & \Lu\LU}\\
    \ket{C_{N-1,3}} &= \ketarray{\Ld\LD & \Ld\LD & \dotsm & \Ld\LU & \Lu\LU}\\
    \ket{C_{N-1,4}} &= \ketarray{\Ld\LD & \Ld\LD & \dotsm & \Ld\LA_1 & \Lu\LU}\\[3mm]
    \ket{C_{N,1}} &= \ketarray{\Ld\LD & \Ld\LD & \dotsm & \Ld\LA_2 & \Lu\LU}\\
    \ket{C_{N,2}} &= \ketarray{\Ld\LD & \Ld\LD & \dotsm & \Ld\LD & \La\LU}\\
    \ket{C_{N,3}} &= \ketarray{\Ld\LD & \Ld\LD & \dotsm & \Ld\LD & \Ld\LU}\\
    \ket{C_{N,4}} &= \ketarray{\Ld\LD & \Ld\LD & \dotsm & \Ld\LD & \Ld\LA_1}
  \end{align}
\end{subequations}
Here, the qudits of the clock register are subdivided into groups of five, where the first three (labeled $\alpha,\beta,\gamma$) represent a logical qutrit and the the last two (labeled $\delta,\epsilon$) represent a logical $4$-qudit.
The logical space is enforced using
\begin{equation}
  \Hlogical = \sum_{i=0}^{N-1} \left((\Hlogical_3)_{\alpha_i\beta_i\gamma_i} + (\Hlogical_4)_{\delta_i\epsilon_i}\right).
\end{equation}
Let $\Ls := (\{\Lu,\La,\Ld\}\times\{\LU, \LA_1, \LA_2, \LD\})^{\times N}$ be the set of labels of logical states, and $\LS := \Span\{\ket{x}\mid x\in \Ls\}$ the vector space spanned by the logical states.
\begin{lemma}\label{lem:Hlogical}
  $\Null(\Hlogical) = \{\ket{a_1,B_1,\dots,a_N,B_N}\mid a_i\in\{\Lu,\La,\Ld\},B_i\in\{\LU,\LA_1,\LA_2,\LD\}\}$.
\end{lemma}
\begin{proof}
  Each summand of $\Hlogical$ acts on a unique subset of qubits $\alpha_i\beta_i\gamma_i$ or $\delta_i\epsilon_i$, and restricts that subset to the logical space by \Cref{lem:HlogicalIV,lem:HlogicalIII}.
\end{proof}
Next, we restrict the logical states to $\ket{C_{i,j}}$ by defining the Hamiltonian $\Hclockp{1}$:\footnote{The left side of this definition gives the logical interpretation of the Hamiltonian terms on the right side.}
\begin{subequations}\label{eq:Hclock1}
  \begin{align}
    &\text{$\LU$ implies $\Lu$ to the right:} &&\quad\sum_{i=1}^{N-1}\sum_{v\in\{a,d\}}\ketbrab{1,v}_{\epsilon_i\alpha_{i+1}} \label{eq:HUu} \\
    &\text{$\Ld$ implies $\Ld$ to the left:} &&+\sum_{i=1}^{N-1}\ketbrab{1,d}_{\beta_{i}\alpha_{i+1}} \label{eq:Hdd}\\
    &\text{$\Ld$ implies $\LD$ to the left:} &&+\sum_{i=1}^{N-1}\sum_{v\in\{u,a_1,a_2\}}\ketbrab{v,1}_{\epsilon_{i}\gamma_{i+1}}\label{eq:HDd}\\
    &\text{$\Lu,\La$ implies $\LU$ to the right:} &&+\sum_{i=1}^{N}\sum_{v\in\{a_1,a_2,d\}}\ketbrab{1,v}_{\beta_i\delta_i} \label{eq:HuaU}\\
    &\text{$\LU$ or $\LA_1$ is last:} &&+\ketbrab{a_2}_{\delta_N}+\ketbrab{d}_{\delta_N} \label{eq:UA1last}
  \end{align}
\end{subequations}
Note, the ``trick'' that makes the first implication \eqref{eq:HUu} work is that $\ket{\LU} = \ket{u,0} + \ket{u',1}$ can be identified using the qubit, so we only need $(2,5)$-projectors.
Further, the $1$-local projectors are not directly allowed in our setting, but they can be implemented with the help of an otherwise unused ancilla qubit.

Ideally, the nullspace of $H_1 := \Hlogical+\Hclockp{1}$ would be spanned by the clock states $\ket{C_{i,j}}$.
Unfortunately, it still contains states of the form $\ket{\dotsm \LA_j\La\dotsm}$ and $\ket{\dotsm\LD\Lu\dotsm}$.
So we obtain the following statement:
\begin{restatable}{lemma}{lemmaCE}\label{lem:CE}
   $\Null(H_1)=\Span(\{\ket{C_{i,j}}\mid i\in[N],j\in[4]\}\cup\{\ket{E_{i,j}}\mid i\in\{2,\dots,N\},j\in[4]\}):=\LS'$ with
  \begin{subequations}\label{eq:valid1}
    \begin{align}
      \ket{C_{1,1}} &= \ket{\Lu\LU}\otimes\ket{\Lu\LU}^{\otimes N-1}\\
      \ket{C_{1,2}} &= \ket{\La\LU}\otimes\ket{\Lu\LU}^{\otimes N-1}\\
      \ket{C_{1,3}} &= \ket{\Ld\LU}\otimes\ket{\Lu\LU}^{\otimes N-1}\\
      \ket{C_{1,4}} &= \ket{\Ld\LA_1}\otimes\ket{\Lu\LU}^{\otimes N-1}\\
      \ket{C_{i,1}} &= \ket{\Ld\LD}^{\otimes i-2}\otimes\ket{\Ld\LA_2,\Lu\LU}\otimes\ket{\Lu\LU}^{\otimes N-i}\\
      \ket{C_{i,2}} &= \ket{\Ld\LD}^{\otimes i-2}\otimes\ket{\Ld\LD,\La\LU}\otimes\ket{\Lu\LU}^{\otimes N-i}\\
      \ket{C_{i,3}} &= \ket{\Ld\LD}^{\otimes i-2}\otimes\ket{\Ld\LD,\Ld\LU}\otimes\ket{\Lu\LU}^{\otimes N-i}\\
      \ket{C_{i,4}} &= \ket{\Ld\LD}^{\otimes i-2}\otimes\ket{\Ld\LD,\Ld\LA_1}\otimes\ket{\Lu\LU}^{\otimes N-i}\\
      \ket{E_{i,1}} &= \ket{\Ld\LD}^{\otimes i-2}\otimes\ket{\Ld\LA_1,\La\LU}\otimes\ket{\Lu\LU}^{\otimes N-i}\\
      \ket{E_{i,2}} &= \ket{\Ld\LD}^{\otimes i-2}\otimes\ket{\Ld\LA_2,\La\LU}\otimes\ket{\Lu\LU}^{\otimes N-i}\\
      \ket{E_{i,3}} &= \ket{\Ld\LD}^{\otimes i-2}\otimes\ket{\Ld\LD,\Lu\LU}\otimes\ket{\Lu\LU}^{\otimes N-i}.
    \end{align}
  \end{subequations}
\end{restatable}
\begin{proof}
  In \Cref{sec:proofs}.
\end{proof}
The ``error terms'' $\ket{E_{i,1}},\ket{E_{i,2}},\ket{E_{i,3}}$ are penalized by the transition terms combining $\ket{C_{i,1}},\dots,\ket{C_{i,4}}$ to $\ket{C_i}$, which are implemented using $\Hclockp{2} = \sum_{i=1}^N \Hclockp{2,i}$, where $\Hclockp{2,i}$ is defined as follows:
\begin{subequations}\label{eq:Hclock2}
  \begin{align}
    &\text{$C_{i,1}\leftrightarrow C_{i,2}$ via $\ket{\LA_2,\Lu}\leftrightarrow \ket{\LD,\La}$:} &&\begin{cases}
      \ketbraa{(\ket{x,0}-\ket{x,1})}_{\alpha_1\beta_1}, &i=1 \\
      \ketbraa{(\ket{a_2,0}-\ket{d,1})}_{\delta_{i-1}\beta_{i}}, &i>1\\
    \end{cases}\label{eq:Ci12}\\
    &\text{$C_{i,2}\leftrightarrow C_{i,3}$ via $\ket{\La,\LU}\leftrightarrow \ket{\Ld,\LU}$:} &&+ \ketbraa{(\ket{a,1}-\ket{d,1})}_{\alpha_i\epsilon_{i}}\label{eq:Ci23}\\
    &\text{$C_{i,3}\leftrightarrow C_{i,4}$ via $\ket{\Ld,\LU}\leftrightarrow \ket{\Ld,\LA_1}$:} &&+ \ketbraa{(\sqrt2\ket{1,u}-\ket{1,a_1})}_{\gamma_i\delta_{i}}\label{eq:Ci34}
  \end{align}
\end{subequations}
In the first constraint \eqref{eq:Ci12}, the first time step is handled as a special case since the first logical qutrit has no preceding qudit.
The transition on the other time steps works because $\ket{a_2,0}_{\delta_i\beta_{i+1}}$ only overlaps with clock state $\ket{C_{i+1,1}}$ ($\LA_2$ is only followed by $\Lu$ and so $\ket{0}_\beta$ overlaps with the $\ket{x,0,0}_{\alpha\beta\gamma}$ term), and $\ket{d,1}_{\delta_i\beta_{i+1}}$ only overlaps with $\ket{C_{i,2}}$ ($\LD$ is only followed by $\La$ or $\Ld$, from which $\ket{1}_{\beta}$ selects $\La$).
Note, the transition $h_{i,i+1}$ from $\ket{C_i}$ to $\ket{C_{i+1}}$ can now be implemented on a single $5$-qudit since $\ket{C_{i,4}}$ and $\ket{C_{i+1,1}}$ only differ in the qudit $\delta_i$.
We define the clock Hamiltonian $\Hclock := \Hlogical + \Hclockp{1} + \Hclockp{2}$.
The transition and selection operators are then defined as follows, where we give two variants of the $C_{\sim i}$ projectors, one acting on a $5$-qudit and the other on a qubit.
\begin{subequations}
  \begin{align}
    h_{i,i+1}(U_Z)&=\frac12 \left(\I_Z \otimes \ketbrab{a_1}_{\delta_i}+\I_Z \otimes \ketbrab{a_2}_{\delta_i} - U_Z^\dagger \otimes \ketbra{a_1}{a_2}_{\delta_i}- U_Z\otimes \ketbra{a_2}{a_1}_{\delta_i}\right)\\
    C_{\le i}^{(5)} &= \ketbra uu_{\delta_i}\\
    C_{\le i}^{(2)} &= \ketbra 11_{\epsilon_i}\\
    C_{\ge i}^{(5)} &= \ketbra dd_{\alpha_i}\\
    C_{\ge i}^{(2)} &= \ketbra 11_{\gamma_i}
  \end{align}
\end{subequations}
Thus, the Hamiltonians from \Cref{thm:generic-clock} can all be implemented as a $(2,5)$-projector.

\subsection{Analysis of the clock Hamiltonian}

In the following, we prove that the nullspace of $\Hclock$ is spanned by the clock states (\Cref{lem:HclockNullSpace}) and $\gamma(\Hclock) = \Omega(1)$ (\Cref{lem:HclockGamma}).
\newcommand{\CS}{\mathcal{C}}
\begin{restatable}{lemma}{lemHclockNullSpace}\label{lem:HclockNullSpace}
  $\Null(\Hclock)=\Span\{\ket{C_1},\dots,\ket{C_N}\}=:\CS$.
\end{restatable}
\begin{proof}
  We can write $\ket{C_i}$ as
  \begin{equation}
    \ket{C_{i}} = \begin{cases}
     (\ket{\Lu\LU} + \ket{\La\LU} + \ket{\Ld\LU} + \ket{\Ld\LA_1})\otimes \ket{\Lu\LU}^{\otimes N-1}, &i=1\\
    \ket{\Ld\LD}^{\otimes i-2}\otimes (\ket{\Ld\LA_2,\Lu\LU} + \ket{\Ld\LD,\La\LU} + \ket{\Ld\LD,\Ld\LU} + \ket{\Ld\LD,\Ld\LA_1})\otimes \ket{\Lu\LU}^{\otimes N-i}, &i>1
    \end{cases}.
  \end{equation}
  One can easily verify that $\Hclock\ket{C_i} =0 $ for all $i\in[N]$.
  Let $\LS' = \Null(\Hlogical + \Hclockp{1})$.
  Then, $\Null(\Hclock) = \LS'\cap\Null(\Hclockp{2})$.
  Let $\ket{\psi}\in\Null(\Hclock)$.
  Since $\ket\psi\in\LS'$, we may write
  \begin{equation}
    \ket\psi = \sum_{i=1}^N\sum_{j=1}^4c_{i,j}\ket{C_{i,j}} + \sum_{i=2}^N\sum_{j=1}^3e_{i,j}\ket{E_{i,j}}
  \end{equation}
  by \Cref{lem:CE}.
  Next, we argue that $e_{i,j} = 0$ for all $i,j$ and $c_{i,1} = c_{i,2} = c_{i,3} = c_{i,4}$ for all $i$.
  Let $i\in\{2,\dots,N\}$ and consider the projector $\Pi = \frac12\ketbraa{(\ket{a,1} - \ket{d,1})}_{\alpha_i\epsilon_i}$ from \Cref{eq:Ci23}.
  Then $\Pi\ket\psi =$
  \begin{subequations}
  \begin{align}
    &\frac{1}{4}(c_{i,2}-c_{i,3})\ket{\Ld\LD}^{\otimes i-1}\otimes (\ket{a00}-\ket{d00})_{\alpha_i\beta_i\gamma_i}\otimes\ket{u'1}_{\delta_i\epsilon_i}\otimes \ket{\Lu\LU}^{\otimes N-i}\label{eq:Pipsi:23}\\
    +&\frac14e_{i,1}\ket{\Ld\LD}^{\otimes i-2}\otimes\ket{\Ld\LA_1}\otimes (\ket{a00}-\ket{d00})_{\alpha_i\beta_i\gamma_i}\otimes\ket{u'1}_{\delta_i\epsilon_i}\otimes \ket{\Lu\LU}^{\otimes N-i}\label{eq:Pipsi:E1}\\
    +&\frac14e_{i,2}\ket{\Ld\LD}^{\otimes i-2}\otimes\ket{\Ld\LA_2}\otimes (\ket{a00}-\ket{d00})_{\alpha_i\beta_i\gamma_i}\otimes\ket{u'1}_{\delta_i\epsilon_i}\otimes \ket{\Lu\LU}^{\otimes N-i}\label{eq:Pipsi:E2}.
  \end{align}
  \end{subequations}
  Note, the three vectors \eqref{eq:Pipsi:23}, \eqref{eq:Pipsi:E1}, \eqref{eq:Pipsi:E2} are pairwise orthogonal and thus $\bra\psi\Pi\ket\psi = \frac1{8}(\abs{c_{i,2}-c_{i,3}}^2 + \abs{e_{i,1}}^2 + \abs{e_{i,2}}^2)$.
  Therefore, $c_{i,2}=c_{i,3}$ and $e_{i,1} = e_{i,2} = 0$.

  Similarly, consider the projector $\Pi = \frac12\ketbraa{(\ket{a_2,0}-\ket{d,1})}_{\delta_{i-1}\beta_i}$ from \Cref{eq:Ci12}. Then $\Pi\ket\psi = $
  \begin{subequations}
  \begin{align}
    &\frac{1}{2\sqrt2}(c_{i,1}-c_{i,2})\ket{\Ld\LD}^{\otimes i-2}\otimes\ket{\Ld}\otimes (\ket{a_20,x00}-\ket{d0,x10})_{\delta_{i-1}\epsilon_{i-1}\alpha_i\beta_i\gamma_i}\otimes\ket{\LU}\otimes \ket{\Lu\LU}^{\otimes N-i}\label{eq:Pipsi:12}\\
    +\, &\frac{1}{2\sqrt2}e_{i,2}\ket{\Ld\LD}^{\otimes i-2}\otimes\ket{\Ld}\otimes (\ket{a_20,a00}-\ket{d0,a10})_{\delta_{i-1}\epsilon_{i-1}\alpha_i\beta_i\gamma_i}\otimes\ket{\LU}\otimes \ket{\Lu\LU}^{\otimes N-i}\label{eq:Pipsi:Ei2}\\
    +\, &\frac{1}{2\sqrt2}e_{i,3}\ket{\Ld\LD}^{\otimes i-2}\otimes\ket{\Ld}\otimes (\ket{a_20,u00}-\ket{d0,u10})_{\delta_{i-1}\epsilon_{i-1}\alpha_i\beta_i\gamma_i}\otimes\ket{\LU}\otimes \ket{\Lu\LU}^{\otimes N-i}\label{eq:Pipsi:Ei3}.
  \end{align}
  \end{subequations}
  Thus $\bra\psi\Pi\ket\psi = \frac14(\abs{c_{i,1}-c_{i,2}}^2 + \abs{e_{i,2}}^2 + \abs{e_{i,3}}^2)$.
  We obtain the additional constraints $c_{i,1}=c_{i,2}$ and $e_{i,3} = 0$.

  Next, we consider the projector $\Pi = \frac13\ketbraa{(\sqrt2\ket{1,u}-\ket{1,a_1})}_{\gamma_i\delta_i}$ from \Cref{eq:Ci34}. Then $\Pi\ket\psi =$
  \begin{subequations}
    \begin{align}
      &\frac1{3\sqrt{2}}(c_{i,3}-c_{i,4})\ket{\Ld\LD}^{\otimes i-1}\otimes\ket{d'01}_{\alpha_i\beta_i\gamma_i}\otimes(\sqrt2\ket{u0}-\ket{a_10})_{\delta_i\epsilon_i}\otimes\ket{\Lu\LU}^{\otimes N-i}\\
      +\,&\frac1{3\sqrt{2}}e_{i+1,1}\ket{\Ld\LD}^{\otimes i-1}\otimes\ket{d'01}_{\alpha_i\beta_i\gamma_i}\otimes(\ket{a_10}-\sqrt2\ket{u0})_{\delta_i\epsilon_i}\otimes \ket{\La\LU}\otimes\ket{\Lu\LU}^{\otimes N-i-1}.
    \end{align}
  \end{subequations}
  Thus, $\bra\psi\Pi\ket\psi = \frac16 (\abs{c_{i,3}-c_{i,4}}^2 + \abs{e_{i+1,1}}^2)$.
  We obtain the additional constraints $c_{i,3} = c_{i,4}$.

  Finally, consider the case $i=1$. Only the case from \Cref{eq:Ci12} differs from above with $\Pi = \frac12\ketbraa{(\ket{x,0}-\ket{x,1})}_{\alpha_1\beta_1}$ and $\Pi\ket\psi=$
  \begin{equation}
    \frac{1}{2\sqrt2}(c_{1,1} - c_{1,2})(\ket{x00}-\ket{x10})_{\alpha_1\beta_1\gamma_1}\otimes\ket{\LU}\otimes\ket{\Lu\LU}^{\otimes N-1}.
  \end{equation}
  This gives constraint $c_{1,1} = c_{1,2}$.
  Hence, $\ket\psi\in\Span\{\ket{C_1},\dots,\ket{C_N}\}$.
\end{proof}

To prove $\gamma(\Hclock)=\Omega(1)$ (\Cref{lem:HclockGamma}), we various corollaries from \cite{GN13}, which follow from Kitaev's Geometric Lemma.

\begin{lemma}\label{lem:gamma(H_1)}
  Let $H_1 := \Hlogical + \Hclockp{1}$. Then $\gamma(H_1) = \Omega(1)$.
\end{lemma}
\begin{proof}
  We apply \Cref{cor:geometric3} with $H_A = \Hlogical$ and $H_B=\Hclockp{1}$.
  $\Hlogical$ and $\Hclockp{1}$ are sums of commuting local terms and thus have $\gamma(\cdot)$ of $\Omega(1)$.
  $\Gamma=\LS\cap{\LS'}^\bot$ is the span of logical states violating one of the $\Hclockp{1}$ constraints (\Cref{eq:Hclock1}).
  It holds that \Cref{lem:CE}, $\bra{u}\Hclockp{1}\ket{v}=0$ for logical vectors $u\ne v$ because $\Hclockp{1}$ is diagonal and logical states can be written as a superposition of computational basis states, such that no computational basis state is used by two different logical states (see \Cref{eq:logical4,eq:logical3}).
  So it suffices to take the minimum over $\ket{v}$ for $v\in\Ls$ such that $v$ violates some projective constraint $\Pi$ from \Cref{eq:Hclock1}.
  Here, $\I-\Pi_B\succeq\Pi$ since $\Hclockp{1}$ is the sum of commuting projectors (including $\Pi$).
  Thus $\braketb{v}{(\I - \Pi_B)} \ge \braketb{v}{\Pi} \ge 1/4$.
  Hence, $\gamma(H_1) = \Omega(1)$.
\end{proof}
\newcommand\wtc{\widetilde{c}}
\newcommand\wtC{\widetilde{C}}

\begin{restatable}{lemma}{lemHclockGamma}\label{lem:HclockGamma}
  $\gamma(\Hclock)=\Omega(1)$.
\end{restatable}
\newcommand{\ES}{\mathcal{E}}
\begin{proof}
  Let $H_A = H_1$, $H_B = \Hclockp{2}$, and apply \Cref{cor:GN:cor2}.
  By \Cref{lem:gamma(H_1)}, we have $\gamma(H_A) = \Omega(1)$.
  Since all projectors of $\Hclockp{2}$ pairwise commute (see \Cref{eq:Hclock2}), $\gamma(H_B)\ge1$.
  The next step is to characterize $\Gamma = \LS'\cap\CS^\bot$.
  Recall $\LS' =\Span\{\ket{C_{i,j}},\ket{E_{i,j}}\}$ and $\CS=\Span\{\ket{C_i}\}$.
  Since $\CS\subseteq \LS'$, we have $\dim(\Gamma) = \dim(\LS') - \dim(\CS) = 6N - 3$.
  First, it is easy to verify $\ket{E_{i,j}}\in\CS^\bot$.
  We complete the $\{\ket{E_{i,j}}\}$ to a basis of $\Gamma$ with
  \begin{subequations}
  \begin{align}
    \ket{\wtC_{i,1}} &= \ket{C_{i,1}} - \ket{C_{i,2}} + \ket{C_{i,3}} - \ket{C_{i,4}} \\
    \ket{\wtC_{i,2}} &= \ket{C_{i,1}} + \ket{C_{i,2}} - \ket{C_{i,3}} - \ket{C_{i,4}} \\
    \ket{\wtC_{i,3}} &= \ket{C_{i,1}} - \ket{C_{i,2}} - \ket{C_{i,3}} + \ket{C_{i,4}}.
  \end{align}
  \end{subequations}
  Hence, $\Gamma = \Span\{\ket{E_{i,j}}, \ket{\wtC_{i,j}}\}$.
  Let $\ket\psi\in\Gamma$ and write
  \begin{subequations}
  \begin{align}
    \ket\psi &= \sum_{i=1}^N\sum_{j=1}^3\wtc_{i,j}\ket{\wtC_{i,j}} + \sum_{i=2}^N\sum_{j=1}^3e_{i,j}\ket{E_{i,j}}\\
    &= \sum_{i=1}^N\bigl((
       \underbrace{\wtc_{i,1}+\wtc_{i,2}+\wtc_{i,3}}_{c_{i,1}})\ket{C_{i,1}}
    + (\underbrace{-\wtc_{i,1}+\wtc_{i,2}-\wtc_{i,3}}_{c_{i,2}})\ket{C_{i,2}}\nonumber\\
    &\qquad+ (\underbrace{\wtc_{i,1}-\wtc_{i,2}-\wtc_{i,3}}_{c_{i,3}})\ket{C_{i,3}}
    + (\underbrace{-\wtc_{i,1}-\wtc_{i,2}+\wtc_{i,3}}_{c_{i,4}})\ket{C_{i,4}}  \bigr) +  \sum_{i=2}^N\sum_{j=1}^3e_{i,j}\ket{E_{i,j}}.
  \end{align}
  \end{subequations}
  By the same arguments used in the proof of \Cref{lem:HclockNullSpace}, we have
  \begin{subequations}
    \begin{align}
      8\bra\psi\Hclockp{2}\ket\psi &\ge \sum_{i=1}^N\bigl(\abs{c_{i,1} - c_{i,2}}^2 + \abs{c_{i,2} - c_{i,3}}^2 + \abs{c_{i,3} - c_{i,4}}^2\bigr) + \sum_{i=2}^N\sum_{j=1}^3\abs{e_{i,j}}^2\\
      &= 4\sum_{i=1}^N\bigl(\abs{\wtc_{i,1}+\wtc_{i,3}}^2 + \abs{\wtc_{i,1}-\wtc_{i,2}}^2 + \abs{\wtc_{i,1}-\wtc_{i,3}}^2\bigr) + \sum_{i=2}^N\sum_{j=1}^3\abs{e_{i,j}}^2\\
      &\ge \frac49\sum_{i=1}^N\sum_{j=1}^3\abs{\wtc_{i,j}}^2 + \sum_{i=2}^N\sum_{j=1}^3\abs{e_{i,j}}^2\ge \frac49\label{eq:psiHclock2psi:last}
    \end{align}
  \end{subequations}
  where \Cref{eq:psiHclock2psi:last} follows because for any $a,b,c\in\CC$, $\abs{a+b} + \abs{a-b}+\abs{a+c}\ge \max\{\abs a, \abs b, \abs c\}$, as $\abs{a+b}+\abs{a-b}\ge \max\{\abs{a},\abs{b}\}$ and thus $\abs{a+b} + \abs{a-b}+\abs{a+c}\ge \abs{a}+\abs{c}-\abs{a}=\abs{c}$.
  Hence $\abs{a+b}^2 + \abs{a-b}^2 + \abs{a+c}^2 \ge \frac13(\abs{a} + \abs{a+b} + \abs{a+c})^2 \ge \frac13\max\{\abs{a}^2,\abs{b}^2,\abs{c^2}\}\ge\frac19 (\abs a^2 + \abs b^2 + \abs c^2)$.

  Therefore, we have $\gamma(H_B|_{\LS'}) = \Omega(1)$.
  Similarly, we can bound $\braketb\psi{\Hclockp{2}} = O(1)$ and thus $\norm{\Hclockp{2}\vert_{\LS'}} = O(1)$, which gives $\gamma(\Hclock) = \Omega(1)$.
\end{proof}

\subsection{\texorpdfstring{$\QMAo$}{QMA1}-completeness}\label{sec:QMA1-completeness}

Our main contribution is certainly the $\QMA_1$-hardness (\Cref{thm:25}), but we still need to discuss containment in $\QMA_1$ briefly.

\begin{lemma}\label{lem:containment}
  The $(2,5)$-QSAT and $(3,4)$-QSAT instances constructed from $\QMA_1$-circuits are contained in $\QMA_1$.
\end{lemma}
\begin{proof}
  To evaluate $(k,l)$-QSAT instances with a $\QMA_1$-verifier, we embed each qudit into $\lceil\log d\rceil$ qubits, and add additional diagonal projectors to reduce the local dimension as necessary.
  The lemma then follows from the containment of QSAT in $\QMA_1$ \cite{GN13}, which requires projectors of a specific form.
  Besides \eqref{eq:Ci34}, all projectors used in \Cref{sec:clock} have entries in $\ZZ[\frac1{\sqrt2},i]$.
  Measurements with respect to such projectors are implemented with \cite{GS13}.
  Recall the projector from \eqref{eq:Ci34},
  \begin{equation}
    \Pi = \ketbraa{\left(\sqrt{\frac23}\ket{1,u}-\sqrt{\frac13}\ket{1,a_1}\right)}_{\gamma_i\delta_{i}}.
  \end{equation}
  Under the qubit embedding, there exists a permutation $P$ such that
  \begin{equation}
    P\Pi P^\dagger = \ketbraa{\left(\sqrt{\frac23}\ket{0000}-\sqrt{\frac13}\ket{0001}\right)}.
  \end{equation}
  A measurement algorithm with perfect completeness is given in \cite[Appendix A]{GN13} for a $3$-local projector of analogous structure, which can easily be extend to larger projectors.
\end{proof}

\thmTwoFive*
\begin{proof}
  Follows from \Cref{thm:generic-clock}, \Cref{lem:HclockNullSpace}, and \Cref{lem:HclockGamma}.
\end{proof}

\thmThreeFour*
\begin{proof}
  We use the same clock construction operating directly in logical space, though care needs to be taken that everything is realizable with $(3,4)$-constraints.
  To implement $\Hclockp{1}$ using only $(3,4)$-terms, we have to replace ``$\Ld$ implies $\Ld$ to the left'' with ``$\LD$ implies $\Ld$ to the left''.
  $\Hclockp{2}$ only uses $(3,4)$-transitions.
  The $C_{\sim i}$ projectors are implemented as $C^{(4)}_{\le i} = \ketbra{U}{U}_{\delta_i}, C^{(3)}_{\le i} = \ketbra uu_{\alpha_i},C^{(4)}_{\ge i}=\ketbra{D}{D}_{\delta_{i-1}},C^{(3)}_{\ge i}=\ketbra{d}{d}_{\alpha_{i}}$, where $\alpha_i,\delta_i$ denote the pairs of $(3,4)$-qudits as depicted in \Cref{eq:valid1}.
  The computational register is implemented on qutrits restricted to a $2$-dimensional subspace.
\end{proof}

\section{Nullspace Connection Lemma}\label{sec:connection-lemma}

The 2D-Hamiltonian, as depicted in \Cref{fig:H}, consists of a sequence of gadgets connected together via the blue zigzag edges.
If we remove the zigzag edges, we can analyze the nullspace and gap of the Hamiltonian easily, as gadgets have only constant size and act on orthogonal clock spaces.
The nullspace of each gadget is spanned by history states in their local clock spaces, as shown in \Cref{lem:conditionHV,lem:conditionHU}.
Each gadget just applies a single gate to the ancilla space, when viewed from the top left corner to the bottom right corner.
Logically, the zigzag edges are just ``identity transitions'' that connect the individual gadgets so that the entire Hamiltonian applies the full circuit from the top left to the bottom right.

When applied to the 2D-Hamiltonian, the clock spaces $K_1,\dots,K_m$ in \Cref{lem:connect} belong to the individual gadget Hamiltonians $H_{1,1},\dots,H_{1,m}$ (on the diagonal of \Cref{fig:H}) with $u_i$ the top left and $v_i$ the bottom right.
The zigzag edges are are modeled via $H_2$ with $V_1=\dotsm=V_{m-1}=\I$.

Although \Cref{lem:connect} is clearly tailored for the 2D-Hamiltonian, it also has broader applicability.
For example, we can also use it to prove the soundness of Kitaev's circuit Hamiltonian \cite{KSV02} without the transformation to the random walk matrix.

We denote the ancilla space by $A$ and the clock space by $C$.

\begin{lemma}[``Nullspace Connection Lemma'']\label{lem:connect}
  Let
  \begin{enumerate}[label=(\arabic*)]
    \item $K_1,\dots,K_m$ be a disjoint partition of the clock states with $u_i,v_i\in K_i,u_i\ne v_i$ for all $i\in[m]$.
    \item $H_{1} = \sum_{i=1}^m H_{1,i}$ be a Hamiltonian such that for all $i\in[m]$:
    \begin{enumerate}
      \item $\Null(H_{1,i}|_{\mathcal{K}_i}) = \Span\{\ket{\psi_i(\alpha_j)}\mid j\in[d]\}$, where $\mathcal{K}_i=\CC^{d}_A\otimes\Span\{\ket{v}_C\mid v\in K_i\}$, and $\ket{\alpha_1},\dots,\ket{\alpha_d}$ is an orthonormal basis of the ancilla space,
      \item $\exists$ linear map $L_i$ with $L_i\ket{\alpha} = \ket{\psi_i(\alpha)}$ and $L_i^\dagger L_i = \lambda_i \I$ for some constant $\lambda_i$,
      \item $H_i$ has support only on clock states $K_i$,
      \item $\norm{\ket{\psi_i(\alpha)}}^2 =: \delta_i \in [1, \Delta]$,
      \item $(\I_A\otimes\bra{u_i}_C)\ket{\psi_i(\alpha)} = \ket{\alpha}_A$,
      \item $(\I_A\otimes\bra{v_i}_C)\ket{\psi_i(\alpha)} = U_i\ket{\alpha}_A$ for some unitary $U_i$.
    \end{enumerate}
    \item $H_2 = \sum_{i=1}^{m-1} h_{v_i,u_{i+1}}(V_i)$ with $h_{v_i,u_{i+1}}(V_i) = \I\otimes\ketbrab{v_i} + \I\otimes\ketbrab{u_{i+1}} - V_i^\dagger \otimes \ketbra{v_i}{u_{i+1}} - V_i \otimes \ketbra{u_{i+1}}{v_i})$ for unitaries $V_i$.
    \item $\ket{\alpha_{ij}} = V_{i-1}U_{i-1}\dotsm V_1U_1\ket{\alpha_j}$.
  \end{enumerate}
  Then for $H=H_1+H_2$, $\Null(H) = \Span\{\sum_{i=1}^m \ket{\psi_i(\alpha_{ij})}\mid j\in[d]\}$
  and $\gamma(H) = \Omega(\gamma(H_1)/(m^2\Delta))$.
\end{lemma}
\begin{proof}
  Since the $H_{1,i}$ Hamiltonians act on different clock spaces, they commute and it holds $S := \Null(H_1) = \Span\{\ket{\psi_i(\alpha_j)}\mid i\in[m],j\in[d]\}$.
  Next, we derive $\Null(H)=S\cap \Null(H_2)$.
  Partition $S$ into orthogonal subspaces $S_1,\dots,S_d$, where $S_j = \Span\{\ket{\psi_i(\alpha_{ij})}\mid i\in[m]\}$.
  Orthogonality holds because $\braket{\psi_i(\alpha)}{\psi_i(\alpha')} = \bra{\alpha}L_i^\dagger L_i\ket{\alpha'} = \lambda_i\braket{\alpha}{\alpha'}$, and $\ket{\psi_i(\alpha)}\in \K_i$, where $\K_i\subseteq\K_{i'}^\bot$ for $i\ne i'$.
  As $H_2$ is block diagonal across $S_1,\dots,S_d$, it suffices to consider the $\Null(H_2\vert_{S_j})$ individually.

  Let $\ket\psi = \sum_{i=1}^m a_{i}\ket{\psi_i(\alpha_{ij})}\in S_j$.
  Then $\braketb{\psi}{H} = \braketb{\psi}{H_2} = \sum_{i=1}^{m-1} \braketb{\psi}{h_{v_i,u_{i+1}}(V_i)}$
  with
  \begin{equation}
    \begin{aligned}
      h_{v_i,u_{i+1}}(V_i)\ket{\psi} &= a_{i}U_i\ket{\alpha_{ij}}\ket{v_i} + a_{i+1}\ket{\alpha_{i+1,j}}\ket{u_{i+1}} - a_{i+1}V_i^\dagger \ket{\alpha_{i+1,j}}\ket{v_i} - a_{i}V_i U_i \ket{\alpha_{ij}}\ket{u_{i+1}}\\
      &=a_{i}U_i\ket{\alpha_{ij}}\ket{v_i} + a_{i+1}\ket{\alpha_{i+1,j}}\ket{u_{i+1}} - a_{i+1}U_i\ket{\alpha_{ij}}\ket{v_i} - a_{i}\ket{\alpha_{i+1,j}}\ket{u_{i+1}}.
    \end{aligned}
  \end{equation}
  Thus, $\bra{\psi}h_{v_i,u_{i+1}}(V_i)\ket{\psi} = a_i^*a_i + a_{i+1}^*a_{i+1} - a^*_{i}a_{i+1} - a^*_{i+1}a_i = \abs{a_i - a_{i+1}}^2$ and $\Null(H_2\vert_{S_j}) = \Span\{\allowbreak\sum_{i=1}^m \ket{\psi_i(\alpha_{ij})}\}$.

  Next, we bound $\gamma(H)$ using \Cref{cor:GN:cor2}, stating $\gamma(H)\ge\min\{\gamma(H_1),\gamma(H_2)\}\cdot \gamma(H_2\vert_S)/(2\norm{H_2})$.
  The terms $\norm{H_2},\gamma(H_2)$ are constant as the $h_{v_i,u_{i+1}}$ projectors act on distinct clock states.
  Hence, $\gamma(H) = \Omega(\gamma(H_1)\gamma(H_2\vert_S))$.
  Since $H_2\vert_S$ is block diagonal, we have $\gamma(H_2\vert_S)\ge\min_{j\in[d]}\gamma(H_2\vert_{S_j})$, where $\gamma(H_2\vert_{S_j}) = \min_{\ket{v}\in \Gamma_j:\braket vv = 1}\braketb{v}{H_2}$ and $\Gamma_j = S_j\cap \Null(H)^\bot$.

  Let $\ket\psi = \sum_{i=1}^m a_i\ket{\psi_i(\alpha_{ij})}\in \Gamma_j, \braketc{\psi}=1$ and $\ket\phi = \sum_{i=1}^m \ket{\psi_i(\alpha_{ij})}\in\Null(H)$.
  Then $0 = \braket{\phi}{\psi} = \sum_{i=1}^m a_i\delta_i$ and $\braketb{\psi}{H_2} = \sum_{i=1}^{m-1} \abs{a_i-a_{i+1}}^2\ge (\sum_{i=1}^{m-1}\abs{a_i-a_{i+1}})^2/m$.
  Also, there exists an $l$ such that $\abs{a_l}^2\delta_l\ge1/m$.
  Via a global rotation, we may assume without loss of generality $a_l>0$.
  Since $\sum_{i=1}^m a_i\delta_i = 0$, there must exist $k$ with $\Re(a_k)<0$.
  Thus $\abs{a_l-a_k}= \abs{a_l\delta_l-a_k\delta_l}/\delta_l>\abs{a_l\delta_l}/\delta_l\ge1/\sqrt{m\Delta}$.
  By the triangle inequality, $\braketb{\psi}{H_2}\ge 1/(m^2\Delta)$.
\end{proof}

\begin{figure}[t]
  \centering
  \ifx\figcommon\undefined\input{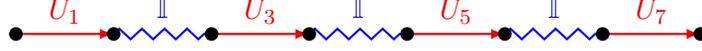}\fi
\renewcommand{\dotSize}{0.07}
\renewcommand{\baseGridW}{8}
\newcommand{\nsteps}{4}
\renewcommand{\baseGridH}{1}

\begin{tikzpicture}[x=1cm,y=1cm,scale=1.3]

\foreach \i in {1,3,...,\baseGridW} {
	\pgfmathsetmacro\ii{int(\i-1)}
	\draw[unitary] (\i,1) -- (\i+1,1) node [midway,above] {$U_\i$};
}
\foreach \i in {2,3,...,\nsteps} {
	\draw[connect] (2*\i-2,1) -- (2*\i-1,1) node [midway,above=2.2pt] {$\mathbb I$};
}

\makeGrid{\baseGridW}{\baseGridH}

\end{tikzpicture}
  \caption{Applying the Nullspace Connection Lemma to Kitaev's circuit Hamiltonian.}
  \label{fig:kitaev}
\end{figure}

\begin{remark}\label{rem:nullspace-kitaev}
  The Nullspace Connection Lemma is quite general and, for example, can also used to prove correctness of Kitaev's original circuit-to-Hamiltonian construction \cite{KSV02} directly without transforming $\hprop$ to a random walk.
  Recall
  \begin{align*}
    \hprop &= \sum_{j=1}^N \Hprop^j,\\
    \hprop^j &= -\frac12 (U_j)_A \otimes \ketbra{j}{j-1}_C - \frac12 (U_j^\dagger)_A\otimes + \frac12I_A\otimes (\ketbra jj + \ketbra{j-1}{j-1})_C,
  \end{align*}
  where $A$ denotes the ancilla space where the computation takes place, and $C$ the clock space.
  Assume that $U_j=\I$ for even $j$.
  \Cref{fig:kitaev} then depicts $\hprop$ for $N = 7$ in the style of \Cref{fig:H}.
  The graph is only a line since $\hprop$ uses an ordinary (single) clock.
  To apply \Cref{lem:connect}, let $K_1=\{0,1\}, K_2=\{2,3\}, K_3 = \{4,5\}, K_4=\{6,7\}$, $H_1 = \sum_{j\in\{1,3,5,7\}}\hprop^j$ (red edges in \Cref{fig:kitaev}), and $H_2 = \sum_{j\in\{2,4,6\}}\hprop^j$ (blue zig-zag edges in \Cref{fig:kitaev}).
  It is now easy to verify that all the conditions of \Cref{lem:connect} are satisfied:
  For example, $\Null(\hprop^1|_{\mathcal{K}_1}) = \Span\{\ket{\psi_1(\alpha_j)}\mid j\in [d]\}$, for $\ket{\psi_1(\alpha_j)} = \ket{\alpha_j}_A\ket{0}_C + U_1\ket{\alpha_j}_A\ket{1}_C$.
  It follows that $\Null(\hprop) = \Span\{\ket{\psihist(\alpha_j)}\mid j\in[d]\}$,
  where $\ket{\psihist(\alpha)} = \sum_{i=0}^N U_{i}\dotsm U_1\ket{\alpha}\ket{i}_C$.

  The Hamiltonian $\hin$ is then used to restrict the ancilla space to $\ket{0}$ on timestep $0$, so that $\Null(\hin+\hprop) = \Span\{\ket{\psihist(0)}\}$.
\end{remark}

\section{(3,d)-QSAT on a line}\label{sec:3d}

\thmThreeD*
To prove this theorem, we give a general construction to embed an arbitrary Hamiltonian $H=\sum_{i=1}^{n-1}H_{i,i+1}$ on $n$ qu$d$its with $d$ into a Hamiltonian $H'$ on an alternating line of $n+1$ qu$d'$its ($d'=(d'')^2$) and $n$ qutrits.
The qu$d'$its are treated as two qu$d''$its of dimension $d'' \in O(d^2)$.
Then each triple $(d'',3,d'')$ logically stores one qu$d$it, which is ``sent'' between the outer qu$d''$its with a history-state-like construction.
Conceptually, we think of this system as three bins with two kinds of balls (say \emph{red} and \emph{black}).
The outer bins (qu$d''$its) may contain up to $B=2d$ balls, and the middle bin only at most a single ball.
A valid configuration of the bins has $B+1$ balls, and the number of red balls is even and positive.
We use transition terms to enforce a superposition of all valid configurations that have the same number of red balls.
Hence, the $(d'',3,d'')$ system has a nullspace of dimension $d$.

\newcommand{\notstandalone}{}
\begin{figure}[t]
  \centering
  \ifx\notstandalone\undefined
\documentclass[tikz,border=1mm]{standalone}
\usepackage{mathtools}
\DeclarePairedDelimiter\ket{\lvert}{\rangle}
\newcommand\p{}
\begin{document}
\fi

\newcommand{\drawBuckets}[3]{
  \draw[bucket] \p{0}{2} -- \p{0}{0} -- \p{2}{0} -- \p{2}{2};
  \node at \p{1}{-0.6} {$\quad#1_{\alpha_\ell}$};
  \draw[] \p{2}{0.25} -- \p{3}{0.25};
  \node at \p{3.5}{-0.6} {$\;\;#2_\beta$};
  \draw[bucket] \p{3}{1} -- \p{3}{0} -- \p{4}{0} -- \p{4}{1};
  \draw[] \p{4}{0.25} -- \p{5}{0.25};
  \draw[bucket] \p{5}{2} -- \p{5}{0} -- \p{7}{0} -- \p{7}{2};
  \node at \p{6}{-0.6} {$\quad#3_{\alpha_r}$};
}

\newcommand{\ballR}[2]{
  \filldraw[red!60] \p{#1+0.5}{#2+0.5} circle (0.35);
}
\newcommand{\ballB}[2]{
  \filldraw \p{#1+0.5}{#2+0.5} circle (0.35);
}

\tikzset{bucket/.style={thick},x=1cm,y=1cm}
{
\newcommand{\x}{0}
\newcommand{\y}{0}
\renewcommand{\p}[2]{(\x+#1,\y+#2)}
\begin{tikzpicture}[scale=0.55]
  \drawBuckets{\ket{2,2}}{\ket2}{\ket{0,0}}
  \ballR{0}{1}
  \ballR{1}{1}
  \ballB{0}{0}
  \ballB{1}{0}
  \ballB{3}{0}

  \renewcommand{\x}{9}
  \drawBuckets{\ket{2,2}}{\ket0}{\ket{0,1}}
  \ballR{0}{1}
  \ballR{1}{1}
  \ballB{0}{0}
  \ballB{1}{0}
  \ballB{5}{0}

  \renewcommand{\x}{18}
  \drawBuckets{\ket{1,2}}{\ket1}{\ket{0,1}}
  \ballR{0}{1}
  \ballB{0}{0}
  \ballB{1}{0}
  \ballR{3}{0}
  \ballB{5}{0}

  \renewcommand{\x}{0}
  \renewcommand{\y}{-4}
  \drawBuckets{\ket{1,2}}{\ket0}{\ket{1,1}}
  \ballR{0}{1}
  \ballB{0}{0}
  \ballB{1}{0}
  \ballB{5}{0}
  \ballR{6}{0}

  \renewcommand{\x}{9}
  \drawBuckets{\ket{0,2}}{\ket1}{\ket{1,1}}
  \ballB{0}{0}
  \ballB{1}{0}
  \ballR{3}{0}
  \ballB{5}{0}
  \ballR{6}{0}

  \renewcommand{\x}{18}
  \drawBuckets{\ket{0,2}}{\ket0}{\ket{2,1}}
  \ballB{0}{0}
  \ballB{1}{0}
  \ballB{5}{0}
  \ballR{6}{0}
  \ballR{5}{1}

\end{tikzpicture}
}

\ifx\notstandalone\undefined
\end{document}
\fi
  \caption{Configurations of the balls and bins for $d=2$. Only configurations in $\C_{1}$ are depicted ($c_{1,1}=2$ red balls and $c_{1,2}+1=3$ black balls). The first two are in $\C_{1}^*$.
  The configurations evolve according to the transitions of \eqref{eq:Hball:trans:1} and \eqref{eq:Hball:trans:2}. There are $d''=15$ possible configurations for the large bins.}
  \label{fig:balls}
\end{figure}
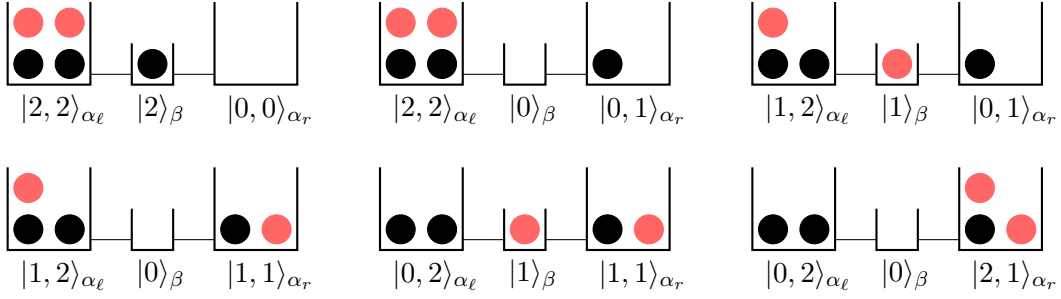

The standard basis states of the qu$d''$it are written as $\ket{c_1,c_2}$ with $c_1,c_2\in\ZZ_{\ge0}$ and $c_1+c_2 \le B := 2d$.
Thus, we get $d''=\sum_{i=0}^{B}(B+1-i)=(B+1)(B+2)/2=(d+1)(2d+1)$.
Semantically, one may think of $c_1$ as ``number of red balls'' and $c_2$ as ``number of black balls'' (see \Cref{fig:balls}).
For $i\in [d]$ let $c_{i,1} = 2i, c_{i,2} = B - 2i$ and define the set of valid configurations corresponding to the state $\ket{i}\in\CC^d$ as
\begin{equation}\label{eq:Ci}
\C_i := \left\{(l_1,l_2,m, r_1,r_2) \;\middle|\;
\scalemath{0.85}{
\setlength{\jot}{0pt}
\begin{aligned}
  l_1,l_2,r_1,r_2&\in\{0,\dots,B\}\\m&\in\{0,1,2\}\\l_1+r_1+\delta_{m,1}&=c_{i,1}\\l_2+r_2+\delta_{m,2}&=c_{i,2}+1
\end{aligned}
}
\right\},
\end{equation}
where $\delta_{x,y}$ denotes the Kronecker delta.
The constraints in \eqref{eq:Ci} enforce that in total there are $c_{i,1}$ red balls and $c_{i,2}+1$ black balls (see also \Cref{fig:balls}).

We place the local terms of the original Hamiltonian into the dimensions corresponding to $\C_i^* :=\{(l_1,l_2,m,r_1,r_2)\in\C_i\mid (l_1,l_2)=c_i^*\vee (r_1,r_2)=c_i^*\}$, where $c_i^*:=(c_{i,1},c_{i,2})$.
These configurations can be identified locally, i.e., one can tell which $\C_i^*$ a configuration corresponds to by only looking at the left or the right bin.
Note that $\abs{\C_i^*} = 4$ and $\C_i^*\cap \C_j=\emptyset$ for $j\ne i$.
Thus, $\bra{c_i^*}_{\alpha_\ell}\ket{\psi_j} =0$ if $j\ne i$.

A logical $\ket{i}$ is then represented by
\begin{equation}\label{eq:psi_i}
  \ket{\psi_i} = \sum_{x=(l_1,l_2,m,r_1,r_2)\in \C_i}\sqrt{w_x}\ket{l_1,l_2}_{\alpha_\ell}\ket{m}_{\beta}\ket{r_1,r_2}_{\alpha_r},\qquad w_x = \begin{cases}
    \frac{2}{3}\cdot \frac{1}{\abs{C_{i}^*}}=:w_i^*, & x\in C_{i}^*\\
    \frac{1}{3}\cdot\frac{1}{\abs{\C_i\setminus \C_i^*}}=:\overline w_i, &x\notin C_{i}^*
  \end{cases},
\end{equation}
where $\alpha_\ell,\alpha_r$ denote the qu$d''$its and $\beta$ the qutrit.
Let
\begin{equation}\label{eqn:V}
    V = \sum_{i=1}^{d} \ketbra{\psi_i}{i}\in\CC^{3(d'')^2\times d}
\end{equation}
be the isometry that maps $\ket{i}$ to $\ket{\psi_i}$.
The weights $w_x$ in \eqref{eq:psi_i} ensure that the $\C_i^*$ always have the same amplitude ($\sqrt{1/6}$), as the $\C_i$ can have different sizes.
The reason for having the additional configurations $\C_i\setminus \C_i^*$ is so that we can use $2$-local transitions (see \eqref{eq:Hball:trans:1} and \eqref{eq:Hball:trans:2}) on the line to enforce a superposition between the $\C_i^*$ states.

Next, we construct a Hamiltonian whose nullspace is spanned by the logical states $\ket{\psi_1},\dots,\ket{\psi_d}$.
Let $\Hball = (\Hballhalf)_{\alpha_\ell\beta} + (\Hballhalf)_{\alpha_r\beta}$, where $\alpha_\ell,\alpha_r$ denote the qu$d''$its and $\beta$ the qutrit.
\begin{subequations}
  \begin{align}
    \Hballhalf &= P(\ket{0,0}\ket{0}) + P(\ket{0,B}\ket{2}) + P(\ket{B,0}\ket{1}) + \sum_{c_1+c_2=B,c_1\text{ odd}} P(\ket{c_1,c_2}\ket{2})\label{eq:Hball:P}\\
    &+\sum_{c_1>0,c_2} \bigl[T((c_1,c_2,0),(c_1-1,c_2,1)) + T((c_1,c_2,2),(c_1-1,c_2+1,1))\bigr]\label{eq:Hball:trans:1}\\
    &+\sum_{c_1,c_2>0} \bigl[T((c_1,c_2,0),(c_1,c_2-1,1)) + T((c_1,c_2,1),(c_1+1,c_2-1,2))\bigr]\label{eq:Hball:trans:2}
  \end{align},
\end{subequations}
where
\begin{subequations}
  \begin{align}
    P(\ket{\psi}) &:= \ketbrab{\psi},\\
    T((a_1,a_2,m_a),(b_1,b_2,m_b)) &:= P(\sqrt{w_{b}}\ket{a_1,a_2}\ket{m_a} - \sqrt{w_{a}}\ket{b_1,b_2}\ket{m_b}),\label{eq:3d-transition}\\
    (w_{a},w_{b}) &= \begin{cases}
      (w_i^*, \overline w_i), &(a_1,a_2) = c_i^*\\
      (\overline w_i, w_i^*), &(b_1,b_2) = c_i^*\\
      (1,1), &\text{otherwise}
    \end{cases}.\label{eq:weights}
  \end{align}
\end{subequations}
One may interpret $P(\ket\psi)$ as penalizing $\ket\psi$ and $T(\cdot)$ as a transition between the two given configurations, where the weights are chosen according to the weights in the $\ket{\psi_i}$ states.

\begin{lemma}
  $\Null(\Hball) = \Span\{\ket{\psi_i}\mid i\in[d]\} =: \Lball$.
\end{lemma}
\begin{proof}
  First, it is easy to see that $\Hball\ket{\psi_i} = 0$ for all $i\in[d]$.
  Next, observe that the transitions of \eqref{eq:Hball:trans:1} and \eqref{eq:Hball:trans:2} force the amplitudes of states corresponding to configurations with the same number of red and black balls to be equal (see the proof of \Cref{lem:HclockNullSpace}), up to scaling by the weights from \eqref{eq:weights}.
  It remains to argue that $\Null(\Hball)$ has no support on ``invalid configurations'' with fewer or more than $B+1$ balls, $0$ red balls, or an odd number of red balls.
  Let $(c_1,c_2)$ be the number of balls in an invalid configuration.
  Since the transition terms enforce a proportional superposition of all configurations of these balls, it suffices to show that there exists a configuration with $c_1$ red balls and $c_2$ black balls that is penalized by \eqref{eq:Hball:P}.
  We have the following cases:
  \begin{enumerate}
    \item $c_1+c_2\le B$. The configuration $\ket{c_1,c_2}\ket{0}\ket{0,0}$ is penalized by $\ket{0,0}\ket{0}$.
    \item $c_2>B$. The configuration $\ket{0,B}\ket{2}\ket{c_1,c_2-B-1}$ is penalized by $\ket{0,B}\ket{2}$.
    \item $c_1>B$. The configuration $\ket{B,0}\ket{1}\ket{B-c_1,c_2}$ is penalized by $\ket{B,0}\ket{1}$.
    \item $c_1\le B,c_2\le B,c_1+c_2\ge B+2$. Then $c_1\ge2$. Let $c_1' := c_1 - 1 + (c_1\bmod 2)$.
    Then $c_1' + c_2 \ge B+1$ and the configuration $\ket{c_1',B-c_1'}\ket{2}\ket{c_1-c_1',c_2-(B-c_1')-1}$ is penalized by the last term of \eqref{eq:Hball:P} because $c_1'$ is odd.
    \item $c_1+c_2=B+1$, $c_1\le B$ odd. Then $\ket{c_1,B-c_1}\ket{2}\ket{0,c_2-(B-c_1)-1}$ is penalized by the last term of \eqref{eq:Hball:P}.
  \end{enumerate}
\end{proof}

The logical states $\ket{\psi_1},\dots,\ket{\psi_d}$ are orthonormal and can be ``identified'' by projecting either qudit onto $\ket{c_i^*}$ as $\bra{c_{i}^*}_{\alpha_r}\ket{\psi_i} = \sqrt{1/6}(\ket{0,0}\ket{2}+\ket{0,1}\ket{0})_{\alpha_\ell\beta}=:\sqrt{1/3}\ket\eta$.
$\ket\eta$ is the residual state of $\ket{\psi_i}$ after projecting onto $\ket{c_i^*}$ and is the same for all $i\in[d]$.
In \Cref{fig:balls}, $\ket\eta_{\alpha_r\beta}$ is the superposition of the right halves of the first two configurations.

We define the isometry $L := \sum_{i=1}^d \ketbra{c_{i}^*}{i}\in\CC^{d''\times d}$ and finally the complete Hamiltonian $H'$ on qu$d'$its $\alpha_0,\dots,\alpha_n$ (each logically divided into two qu$d''$its $\gamma_i$ and $\delta_i$) and qubits $\beta_1,\dots,\beta_n$ is given by
\begin{subequations}
\begin{align}
  H' &= \Hlog + \Hsim\label{eqn:iso1}\\
  \Hlog &= \sum_{i=1}^n \bigl(\underbrace{(\Hballhalf)_{\delta_{i-1}\beta_{i}} + (\Hballhalf)_{\gamma_i\beta_{i}}}_{(\Hball)_{\delta_{i-1}\beta_i\gamma_{i}}}\bigr)\label{eqn:iso2}\\
  \Hsim &=\sum_{i=1}^{n-1} (L\otimes L)(H_{i,i+1})(L\otimes L)^\dagger_{\alpha_i},\label{eqn:iso3}
\end{align}
\end{subequations}
where $\Hlog$ contains the terms of $\Hball$ to enforce the logical subspace, and $\Hsim$ embeds the terms of the original Hamiltonian $H$.
\Cref{fig:H'} gives a graphical representation of $H'$.
Note that $H'$ acts as identity on the first half of the first qu$d'$it ($\gamma_0$) and the second half of the last qu$d'$it ($\delta_n$).
So we can write
\begin{equation}\label{eqn:H''}
    H' = \I_{\gamma_0}\otimes H''_{\delta_0\alpha_1\beta_1\dots\alpha_{n-1}\beta_n\gamma_{n}} \otimes \I_{\delta_n}.
\end{equation}
The next lemma shows that $H'$ and $H$ are equal inside the nullspace of $\Hlog$, up to an isometry.

\begin{figure}[t]
  \centering
  \input{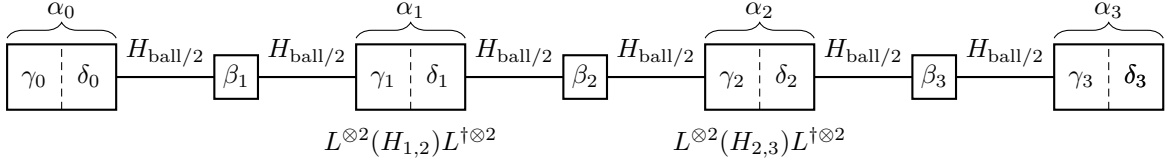}
  \caption{Graphical representation of $H'$ embedding an $n=3$ qudit Hamiltonian.
  The qu$d'$its $\alpha_0,\dots,\alpha_3$ are subdivided into qu$d''$its $\gamma_i,\delta_i$.
  $H'$ acts trivially on $\gamma_0,\delta_3$.}
  \label{fig:H'}
\end{figure}

\begin{lemma}\label{lem:H'L}
  $T^\dagger H'' T = \frac19 H$ with $T = V^{\otimes n}$, for $V$ defined in \Cref{eqn:V} and $H''$ in \Cref{eqn:H''}.
\end{lemma}
\begin{proof}
  It suffices to check equality for basis states, so let $x\in[d]^n$.
  Then $T\ket{x} = \ket{\psi_{x_1}}\dotsm\ket{\psi_{x_n}} =: \ket{\psi_x}$.
  Clearly, $\Hlog \ket{\psi_x} = 0$.
  Consider now the first summand of $\Hsim$, $(L\otimes L)(H_{1,2})(L\otimes L)^\dagger_{\alpha_i} =: M_1$.
  We have
  \begin{equation}
    V^{\dagger\otimes 2}M_1 \ket{\psi_{x_1}}_{\delta_0\beta_1\gamma_1}\ket{\psi_{x_2}}_{\delta_1\beta_2\gamma_2} = \frac13V^{\dagger\otimes 2}\ket{\eta}_{\delta_0\beta_1}\otimes L^{\otimes 2}H_{1,2}\ket{x_1}_{\epsilon_1}\ket{x_2}_{\epsilon_2}\otimes\ket{\eta}_{\gamma_2\beta_2}
    = \frac19H_{1,2}\ket{x_1x_2},
  \end{equation}
  where $\epsilon_1,\epsilon_2$ denote the qudits $H_{1,2}$ acts on, and the second equality follows from the fact that $V^\dagger(\ket\eta_{\alpha_\ell\beta}\otimes L_{\alpha_r}) = \I/\sqrt{3}$.
  Since $M_1$ only acts nontrivially on the first two logical qudits, we have
  $T^\dagger M_1 T \ket x = \frac19 H\ket{x}$.
  Applying the same argument to the other summands of $\Hsim$ yields $T^\dagger H''T \ket{x} = \frac19 H\ket{x}$.
\end{proof}

\begin{lemma}\label{lem:H'gamma}
  Let $H$ be a Hamiltonian on a line of $n$ qu$d$its with $H_{i,i+1}\succeq 0$ and $\gamma(H_{i,i+1})\in \Omega(1)$ for all $i\in[n-1]$.
  There exists an efficiently computable Hamiltonian $H'$ on an alternating chain of $n+1$ qu$d'$its and $n$ qutrits with $d' = ((d+1)(2d+1))^2\in O(d^4)$, such that $\lmin(H')=0$ iff $\lmin(H)=0$ and $\gamma(H')\in\Omega(\gamma(H)/\norm{H})$.
\end{lemma}
\begin{proof}
  If there exists $\ket\psi$ such that $H\ket\psi = 0$, then $H''T\ket{\psi}=0$ by \Cref{lem:H'L}.
  To prove $\gamma(H')\in\Omega(\gamma(H)/\norm{H})$, apply \Cref{cor:GN:cor2}.
  We have $\Null(\Hlog) = \Lball^{\otimes n} =: S$ and $\gamma(\Hlog),\gamma(\Hsim)\in\Omega(1)$ since $\Hlog,\Hsim$ are sums of commuting Hamiltonians.
  Since $T^{\dagger}\Pi_S = T^{\dagger}$, it follows that $T^\dagger H'' T=\frac19 H$ is equal to $H''\vert_S$ up to change of basis.
  Hence $\gamma(H'\vert_S) = \frac19\gamma(H)$ and $\norm{H'\vert_S} = \frac19\norm{H}$.
\end{proof}

\begin{proof}[Proof of \Cref{thm:3d}]
  2-QSAT one a line of qu$d$its with $d=11$ is $\QMAo$-complete \cite{nagajLocalHamiltoniansQuantum2008}.
  Using \Cref{lem:H'gamma}, we can embed this QSAT instance into a $(3,d')$-QSAT on a line.
\end{proof}

\begin{remark}
  Care needs to be taken regarding containment in $\QMAo$, since the transition terms of $\Hballhalf$ include irrational amplitudes (see \eqref{eq:psi_i} and \eqref{eq:3d-transition}) for which the techniques of \cite{GN13} (see \Cref{lem:containment}) do not directly apply.
  If we allow the $\QMAo$-verifier to use gates with entries from some algebraic field extension (as in \cite{bravyiEfficientAlgorithmQuantum2006}), we can easily verify $H'$.
  If we are restricted to the ``Clifford + T'' gate set as in \Cref{def:QMA1}, we can modify $\Hballhalf$ so that the sets $\C_i$ are of equal size.
  We can do this by adding additional dimensions with transitions to $\ket{c_i^*}$, which increases $d'$ to $O(d^6)$.
  Then we can set all weights in the transitions of $\Hball$ to $1$ so that the logical states $\ket{\psi_i}$ are just uniform superpositions over the $\C_i$ configurations.
  \Cref{lem:H'L} then still holds, albeit with a smaller factor that depends on $d$.
\end{remark}

\begin{remark}[Hamiltonian simulation]
  Our embedding of the Hamiltonian $H$ into $H'$ is related to the notion of Hamiltonian simulation \cite{BH17,cubittUniversalQuantumHamiltonians2018}.
  In a sense, our embedding is almost a \emph{perfect simulation} \cite[Definition 20]{cubittUniversalQuantumHamiltonians2018}, but then only the nullspace is really simulated perfectly since $\Hlog$ does not commute with $\Hsim$.
  Our construction takes the form $H' = \Hlog+ \Hsim$, such that $T^\dagger \Hsim T = cH$ for some constant $c$, where $T$ is an isometry with $TT^\dagger = \Pi_{\Null(\Hlog)}$.
  This notion of simulation may be helpful for future quantum SAT research.
\end{remark}

\newcommand\ST{ * {} {}}
\newcommand\bS{{\bar S}}
\section{Hamiltonian with unique entangled ground state on a (2,4)-Line}\label{sec:line}

We were only able to prove $\QMA_1$-hardness of quantum SAT on a $(3,d)$-line.
This raises the question whether hardness still holds with qubits instead of qutrits.
In this section, we show that it is at least possible to construct a $(2,4)$-QSAT instance on a line with a unique entangled null state.
Therefore, the $(2,d)$-QSAT problem on a line does not necessarily have a product state solution like $2$-QSAT.

{\renewcommand\footnote[1]{}\thmTwoFourLine*}

\begin{remark}
  Observe that besides the left and right boundary (projectors $L,R$), all $(4,2)$-projectors have the same form $A$ and all $(2,4)$-projectors have the same form $B$.
  Therefore, $H$ may be considered translation invariant, in a weaker sense.
  The Hamiltonian with a fully entangled ground state on a line of qutrits \cite{bravyiCriticalityFrustrationQuantum2012} also has additional projectors on the boundary.
\end{remark}

We begin by explicitly writing the unique ground state of this Hamiltonian on line of $6$ particles.
The dimensions of these particles are $(4,2,4,2,4,2)$, although for the first and second to last particle a qutrit would suffice as the the $\ket0/\ket{3}$ dimension is not used.
{
\addtolength{\jot}{-1mm}
\begin{equation}
\begin{aligned}
     &\ket{\texttt{ 1 0 | 0 0 | 0 0 }}\\
  +\;&\ket{\texttt{ 2 1 | 0 0 | 0 0 }}\\
  +\;&\ket{\texttt{ 2 0 | 1 0 | 0 0 }}\\
  +\;&\ket{\texttt{ 3 1 | 1 0 | 0 0 }}\\
  +\;&\ket{\texttt{ 2 0 | 2 1 | 0 0 }}\\
  +\;&\ket{\texttt{ 3 1 | 2 1 | 0 0 }}\\
  +\;&\ket{\texttt{ 2 0 | 2 0 | 1 0 }}\\
  +\;&\ket{\texttt{ 3 1 | 2 0 | 1 0 }}\\
  +\;&\ket{\texttt{ 2 0 | 3 1 | 1 0 }}\\
  +\;&\ket{\texttt{ 3 1 | 3 1 | 1 0 }}\\
  +\;&\ket{\texttt{ 2 0 | 2 0 | 2 1 }}\\
  +\;&\ket{\texttt{ 3 1 | 2 0 | 2 1 }}\\
  +\;&\ket{\texttt{ 2 0 | 3 1 | 2 1 }}\\
  +\;&\ket{\texttt{ 3 1 | 3 1 | 2 1 }}
\end{aligned}
\end{equation}
}
While this state might look complex at a first glance, it can easily be understood semantically.
We again think of particles as bins holding balls, but now there is only one kind of ball.
A state $\ket{c}$ means that the bin holds $c$ balls.
Thus, a qu$d$it can hold at most $d-1$ balls and we have \emph{large} bins of capacity $3$ , and \emph{small} bins of capacity $1$.
Initially, only the first bin contains a ball (first state in the superposition).
Then we evolve according to the following rules (also in the reverse):
\begin{itemize}
  \item If a large bin is not empty and the bin to the right is empty, we can add a ball to both bins ($\ket{c,0}\leftrightarrow \ket{c+1,1}$ for $c\in[1,d-2]$).
  \item If a small bin contains a ball and the large bin to the right is empty, we can move the ball from the small to the large bin ($\ket{1,0}\leftrightarrow \ket{0,1}$).
\end{itemize}
We can simplify the transitions by factoring $\ket{*} := \sqrt{1/2}(\ket{20} + \ket{31})$.
To obtain a uniform superposition, we also need to change the amplitudes of the transitions.
On $8$ particles, we obtain the following state (not normalized here).
{
\addtolength{\jot}{-1mm}
\begin{equation}\label{eq:ball-states}
\ket\phi  := \qquad
\begin{aligned}
     &\ket{\text{\texttt{ 1  0 | 0 0 | 0 0 | 0 0 }}}\\
  +\;&\ket{\texttt{ 2  1 | 0 0 | 0 0 | 0 0 }}\\
  +\;&\ket{\texttt{ \ST  | 1 0 | 0 0 | 0 0 }}\\
  +\;&\ket{\texttt{ \ST  | 2 1 | 0 0 | 0 0 }}\\
  +\;&\ket{\texttt{ \ST  | \ST | 1 0 | 0 0 }}\\
  +\;&\ket{\texttt{ \ST  | \ST | 2 1 | 0 0 }}\\
  +\;&\ket{\texttt{ \ST  | \ST | \ST | 1 0 }}\\
  +\;&\ket{\texttt{ \ST  | \ST | \ST | 2 1 }}
\end{aligned}
\end{equation}
}
Note, $\ket\phi$ has a quite similar structure to common clock constructions and can be extended to $2n$ particles.
We will show that it is the unique ground state of the following Hamiltonian:
\begin{align}
  H &= \ketbra00_1 + \ketbra33_{2n-1} + \sum_{i=1}^n A_{2i-1,2i} + \sum_{i=1}^{n-1} B_{2i,2i+1} \\
  A &= \bigl(\ket{10} - \ket{21}\bigr)\bigl(\bra{10}-\bra{21}\bigr) + \bigl(\ket{20} - \ket{31}\bigr)\bigl(\bra{20}-\bra{31}\bigr) + \ketbra{30}{30} + \ketbra{11}{11} \\
  B &= \bigl(\ket{10} - \sqrt{2}\ket{01}\bigr)\bigl(\bra{10}-\sqrt{2}\bra{01}\bigr)
\end{align}

\begin{lemma}\label{eq:phi}
  $\ket\phi$ is fully entangled, i.e. $\ket\phi \ne \ket{\phi_1}_A\otimes \ket{\phi_2}_B$ for all cuts $A/B$ and $\ket{\phi_1},\ket{\phi_2}$.
\end{lemma}
\begin{proof}
  Consider the random experiment of measuring $\ket\psi$ in standard basis.
  The outcome is denoted by the string $x$.
  Let $S\subset [2n]$ be a subset of particles.
  If $\ket\phi = \ket{\phi_S}\ket{\phi_\bS}$, then the random variables $x_S$ and $x_\bS$ (substrings of $x$ on particles $S$ and $\bS=[2n]\setminus S$, respectively) are independent.
  In the following, we argue that this is not the case.

  Note for an odd $i$, $P(x_i=3,x_{i+1}=0)=0$, but $P(x_i=3)P(x_{i+1}=0)>0$.
  Hence, if there exists odd $i$ such that $\abs{\{i,i+1\} \cap S}=1$, then
  $x_S, x_\bS$ are not independent.

  Otherwise, there exists an odd $i$ such that $\abs{\{i,i+2\}\cap S}=1$.
  Again, $x_S, x_\bS$ are not independent as $P(x_i=0,x_{i+2}=1)=0$, but $P(x_i=0)P(x_{i+2}=1)>0$.
\end{proof}

\begin{lemma}
  $\ket\phi$ is the unique ground state of $H$.
\end{lemma}
\begin{proof}
  It is easy to verify $H\ket\phi=0$.
  Now, assume $H\ket\psi=0$.
  If there exists a standard basis vector $\ket x$ such that $\braket x\psi \ne 0$, it corresponds to an illegal state of the ball game (terms of \eqref{eq:ball-states} are the legal states).
  Observe that the transition terms of $A$ and $B$ directly correspond to the allowed moves in the ball game.
  The illegal states that are not caught directly, are \texttt{|10|} or \texttt{|21|} not followed by all zeroes.
  By applying the transition rules, we can go to \texttt{|11|}, which is caught by $A$.
  Hence, $\ket\psi$ and $\ket\phi$ overlap the same standard basis vectors.
  The transition terms ensure the weights are such that $\ket\phi$ can be written as in \eqref{eq:phi}.
\end{proof}

\begin{remark}
  $\ket\phi$ has only constant entanglement entropy, whereas \cite{bravyiCriticalityFrustrationQuantum2012} achieves $\Omega(\log n)$.
  So it remains an open problem whether logarithmic entanglement entropy can be achieved on the $(2,4)$-line.
\end{remark}


\printbibliography

\appendix

\section{Omitted Proofs}\label{sec:proofs}

The proofs of the following lemmas are given in the appendix since they are purely mechanical.

\lemmaHlogicalIII*
\begin{proof}
  It is trivial to verify that $\ket{\Lu},\ket{\La},\ket{\Ld}$ are indeed in the nullspace of $\Hlogical_3$.
  Expanding the implicit identities of $\Hlogical_3$ in the standard basis, we can write $\Hlogical_3$ as the sum of rank-$1$ projectors:
  \begin{equation}\label{eq:Hlogical3-vectors}
    \begin{aligned}
      &\ket{u01},\ket{u00},\gray{\ket{u01}},\ket{u11},\\
      &\ket{a10},\ket{a11},\ket{a01},\gray{\ket{a11}},\\
      &\ket{d10},\ket{d11},\ket{d01},\gray{\ket{d11}},\\
      &\ket{d'10},\ket{d'11},\ket{d'00},\gray{\ket{d'10}},\\
      &\ket{x01},\ket{x11},\\
      &\ket{x00}-\ket{u10},\gray{\ket{x01}-\ket{u11}},\\
      &\ket{x10}-\ket{a00},\gray{\ket{x11}-\ket{a11}}\\
      &\ket{d00}-\ket{d'01},\gray{\ket{d10}-\ket{d'11}}
    \end{aligned}
  \end{equation}
  The gray terms mark vectors in the span of prior vectors.
  The remaining $17$ terms form an orthogonal basis of $\Image(\Hlogical_3)$.
\end{proof}

\lemmaCE*
\begin{proof}
  By \Cref{lem:Hlogical}, we know $\Null(\Hlogical + \Hclockp{1}) = \Null(\Hclockp{1})\cap \L = \Null(\Pi_\LS \Hclockp{1} \Pi_\LS) \cap \L$.
  $\Pi_\LS \Hclockp{1} \Pi_\LS$ then consists of the terms $\ketbra\psi\psi\otimes\I$, where $\ket\psi$ is one of the following vectors, and $x_i:=\alpha_i\beta_i\gamma_i, Y_i:=\delta_i\epsilon_i$ denote the $i$-th logical qutrit and $4$-qudit, respectively:
  \begin{subequations}\label{eq:pairs}
    \begin{align}
      &\ket{\LU\La}_{Y_ix_{i+1}},\ket{\LU\Ld}_{Y_ix_{i+1}},\label{eq:vUu}\\
      &\ket{\Ld\Lu}_{x_ix_{i+1}},\ket{\Ld\La}_{x_ix_{i+1}},\\
      &\ket{\LU\Ld}_{Y_ix_{i+1}},\ket{\LA_1\Ld}_{Y_ix_{i+1}},\ket{\LA_2\Ld}_{Y_ix_{i+1}},\\
      &\ket{\Lu\LA_1}_{x_iY_i},\ket{\La\LA_1}_{x_iY_i},\ket{\Lu\LA_2}_{x_iY_i},\ket{\La\LA_2}_{x_iY_i},\ket{\Lu\LD}_{x_iY_i},\ket{\La\LD}_{x_iY_i}\label{eq:vuaD}\\
      &\ket{\LA_2}_{Y_N},\ket{\LD}_{Y_N}\label{eq:vlast}
    \end{align}
  \end{subequations}
  Note, cross terms do not appear because $\Hclockp{1}$ is diagonal and logical states can be written as a superposition of computational basis states, such that no computational basis state is used by two different logical states (see \Cref{eq:logical4,eq:logical3}).
  The vectors from Equations \eqref{eq:HUu} to \eqref{eq:UA1last} correspond to the projectors from Equations \eqref{eq:vUu} to \eqref{eq:vlast}, respectively.
  Also note, the pairs not listed in \Cref{eq:pairs} are precisely those satisfying the conditions written on the left side of \Cref{eq:Hclock1}.
  Since $\Pi_\LS \Hclockp{1} \Pi_\LS$ is diagonal in the logical basis, its nullspace must also be spanned by logical basis states.

  To prove this lemma, first observe that all states in \Cref{eq:valid1} satisfy the conditions, or equivalently, contain no ``illegal pair'' from \Cref{eq:pairs}.
  Let $\ket\psi=\ket{x_1Y_1\dots x_NY_N}\in\Null(\Pi_\LS \Hclockp{1} \Pi_\LS) \cap \LS$ be a logical basis state.
  Next, we argue that only states $\ket\psi$ in \Cref{eq:valid1} can satisfy the conditions from \Cref{eq:Hclock1} using a case distinction.
  \begin{enumerate}
    \item $\forall i: x_i\ne\Ld$. Then $\ket\psi\in\{\ket{C_{1,1}},\ket{C_{1,2}}\}$ by \eqref{eq:HUu} and \eqref{eq:HuaU}.
    \item $\exists i:x_i=\Ld$. Let $i$ be maximal such that $x_i=\Ld$. Then $x_1Y_1\dotsm x_{i-1}Y_{i-1}=(\Ld\LD)^{i-1}$ by \eqref{eq:Hdd} and \eqref{eq:HDd}.
    \begin{enumerate}
      \item $Y_i=\LU$. Then $Y_ix_{i+1}Y_{i+1}\dotsm x_NY_N=\LU(\Lu\LU)^{N-i}$ by \eqref{eq:HUu} and \eqref{eq:HuaU}. Hence $\ket{\psi} = \ket{C_{i,3}}$.
      \item $Y_i=\LA_1$. Note, $x_{i+1}\ne \Ld$ by choice of $i$.
      \begin{enumerate}
        \item $i=N$. Then $\ket\psi = \ket{C_{N,4}}$.
        \item $x_{i+1}=\Lu$. Then $x_{i+1}Y_{i+1}\dotsm x_NY_N=(\Lu\LU)^{N-i}$ by \eqref{eq:HUu} and \eqref{eq:HuaU}. Hence $\ket\psi = \ket{C_{i,4}}$.
        \item $x_{i+1}=\La$. Then $x_{i+1}Y_{i+1}\dotsm x_NY_N=\La\LU(\Lu\LU)^{N-i-1}$ by \eqref{eq:HUu} and \eqref{eq:HuaU}. Hence, $\ket\psi = \ket{E_{i+1,1}}$.
      \end{enumerate}
      \item $Y_i=\LA_2$. Again, $x_{i+1}\ne \Ld$ by choice of $i$. $i<N$ by \eqref{eq:UA1last}.
      \begin{enumerate}
        \item $x_{i+1}=\Lu$. Then $x_{i+1}Y_{i+1}\dotsm x_NY_N=(\Lu\LU)^{N-i}$ by \eqref{eq:HUu} and \eqref{eq:HuaU}. Hence $\ket\psi = \ket{C_{i+1,1}}$.
        \item $x_{i+1}=\La$. Then $x_{i+1}Y_{i+1}\dotsm x_NY_N=\La\LU(\Lu\LU)^{N-i-1}$ by \eqref{eq:HUu} and \eqref{eq:HuaU}. Hence, $\ket\psi = \ket{E_{i+1,2}}$.
      \end{enumerate}
      \item $Y_i=\LD$. Again, $x_{i+1}\ne \Ld$ and $i<N$.
      \begin{enumerate}
        \item $x_{i+1}=\Lu$. Then $x_{i+1}Y_{i+1}\dotsm x_NY_N=(\Lu\LU)^{N-i}$ by \eqref{eq:HUu} and \eqref{eq:HuaU}. Hence $\ket\psi = \ket{C_{i+1,2}}$.
        \item $x_{i+1}=\La$. Then $x_{i+1}Y_{i+1}\dotsm x_NY_N=\La\LU(\Lu\LU)^{N-i-1}$ by \eqref{eq:HUu} and \eqref{eq:HuaU}. Hence, $\ket\psi = \ket{E_{i+1,3}}$.
      \end{enumerate}
    \end{enumerate}
  \end{enumerate}
\end{proof}

\section{Intuition behind the clock}\label{sec:intuition}

To give the intuition behind our clock construction for $(2,5)$-QSAT, we first sketch how to obtain the $\QMAo$-hardness for $(2,7)$-QSAT by ``simulating'' the $(3,5)$-QSAT construction of \cite{ER08}.

On qutrits, the following clock construction is straightforward to implement with two-local constraints on a line.
\begin{equation}\label{eq:clock33}
    \begin{aligned}
    \cdots\texttt{1000}\cdots\\
    \cdots\texttt{2000}\cdots\\
    \cdots\texttt{2100}\cdots\\
    \cdots\texttt{2200}\cdots\\
    \cdots\texttt{2210}\cdots\\
    \cdots\texttt{2220}\cdots\\
    \cdots\texttt{2221}\cdots\\
    \cdots\texttt{2222}\cdots
    \end{aligned}
\end{equation}
Observe that only a single qutrit changes in each transition.
The same clock can be implemented on a $(2,6)$ system by using the following logical qutrits on a $6$-qudit $\alpha$ and three qubits $\beta,\gamma,\delta$.
\begin{equation}\label{eq:logic33}
    \begin{aligned}
    \ket{\mathbf{0}} &= \ket{0}_\alpha\ket{000}_{\beta\gamma\delta} + \ket{0'}_\alpha\ket{100}_{\beta\gamma\delta}\\
    \ket{\mathbf{1}} &= \ket{1}_\alpha\ket{000}_{\beta\gamma\delta} + \ket{1'}_\alpha\ket{010}_{\beta\gamma\delta}\\
    \ket{\mathbf{2}} &= \ket{2}_\alpha\ket{000}_{\beta\gamma\delta} + \ket{2'}_\alpha\ket{001}_{\beta\gamma\delta}
    \end{aligned}
\end{equation}
Here, $\alpha$ is treated as if it were a qubit and a qutrit.
It is straightforward to construct a $2$-local Hamiltonian $H_{\rm dit}$ on the system $\alpha \beta\gamma\delta$ whose nullspace is spanned by $\ket{\mathbf0},\ket{\mathbf1},\ket{\mathbf2}$.

Observe that each logical qutrit can be identified by a single qubit and transitions only require the qudit $\alpha$.
Thus, we can construct $2$-local $H_{\rm clock}$ whose nullspace is spanned by valid logical clock states \eqref{eq:clock33}.

A $\mathbf{10} \leftrightarrow \mathbf{20}$ transition can be implemented as
\begin{equation}
    \bigl(\ket{1}-\ket{2}\bigr)\bigl(\bra{1}-\bra{2}\bigr)_{\alpha_i}\otimes \ketbra11_{\beta_{i+1}},
\end{equation}
where register $x_i$ belongs to logical qudit $i$.
The remaining transitions can be implemented analogously, which allows us to define a transition Hamiltonian $H_{\rm trans}$ enforcing a uniform superposition over all time steps.

It is straightforward to extend this construction to qudits of any dimension.
This way, we can also represent the states $u,a_1,a_2,a_3,d$ from \cite{ER08}, which gives QMA$_1$ hardness.
We remark that transitions between $a_1,a_2,a_3$ can be done $1$-locally because they only occur once in each logical dit over all time steps.
Therefore, we also only need the left summand in \eqref{eq:logic33} for the logical states corresponding to $a_1,a_2,a_3$, and the summands starting with $\ket{x'}_\alpha$ for $x\in\{a_1,a_2,a_3\}$ are not necessary.
Hence, a $(2,7)$ system suffices.

\end{document}